\documentclass[11pt]{article}

\usepackage[top=1in,bottom=1in,left=1.5in,right=1.5in]{geometry}
\usepackage{amsthm}
\usepackage{mathtools}
\usepackage{amssymb}
\usepackage{hyperref}
\usepackage{xspace}
\usepackage{subfig}
\usepackage{bm}
\usepackage{authblk}
\usepackage{microtype}

\usepackage{stmaryrd}
\SetSymbolFont{stmry}{bold}{U}{stmry}{m}{n}
 
\usepackage{graphicx}
\graphicspath{{resources/}}

\usepackage{tikz}
\usetikzlibrary{arrows, shapes.gates.logic.US, calc}
\tikzstyle{branch}=[fill, shape=circle, minimum size=3pt, inner sep=0pt]

\input{macros}

\newtheorem{theorem}{Theorem}
\newtheorem{lemma}[theorem]{Lemma}
\newtheorem{construction}[theorem]{Construction}

\begin{document}
    
\title{
    Robust Chemical Circuits\footnote{
        This work is supported by National Science Foundation grants 1247051 and 1545028.
        A preliminary version of a portion of this work was presented at the Sixth International Conference on the Theory and Practice of Natural Computing (TPNC 2017, Prague, Czech Republic, December 18--20, 2017).
    }
}

\author[1]{Samuel J. Ellis}

\author[2]{Titus H. Klinge}

\author[3]{James I. Lathrop}
\affil[1]{The Molecular Sciences Software Institute, Blacksburg, VA 24060}
\affil[2]{Carleton College, Northfield, MN 55057}
\affil[3]{Iowa State University, Ames, IA 50011}

\date{}

\maketitle

\vspace{-1em}

\begin{abstract}
    We introduce a new motif for constructing robust digital logic circuits using input/output chemical reaction networks.
    These chemical circuits robustly handle perturbations in input signals, initial concentrations, rate constants, and measurements.
    In particular, we show that all combinatorial circuits and several sequential circuits enjoy this robustness.
    Our results compliment existing literature in the following three ways:
    (1) our logic gates read their inputs catalytically which make ``fanout'' gates unnecessary;
    (2) formal requirements and rigorous proofs of satisfaction are provided for each circuit; and
    (3) robustness of every circuit is closed under modular composition.
\end{abstract}

\section{Introduction}\label{sec:intro}
The development of affordable, fast, and reliable electronic logic circuits has broadly impacted society by accelerating many scientific and technological advancements.
Similarly, biochemical logic circuits have potential to broadly impact methods in drug therapy, bio-diagnostics, and synthetic biology.
Research into biochemical circuits dates back at least to~\cite{jHjWeRo91}, and since then many theoretical motifs for implementing logic circuits have been proposed~\cite{jMagn97,jHFLD09,jiang13,jGZWYZ17,oElli17,beiki18,arkin94,jQiaWin11a,garg18}.

Chemical reaction networks (CRNs) are currently the mathematical model of choice for biochemical computing and have been studied for over 50 years~\cite{jAris65}.
This is primarily due to recent results showing they are computationally powerful~\cite{oCSWB09,jSCWB08,cFLBP17} and can be implemented using DNA molecules~\cite{jSoSeWi10,jCard13,jCDSPCS13,srinivas17,cBSJDTW17} using toehold-mediated strand displacement~\cite{jYTMSN00,jZhaWin09,jZhaSee11,jLYCP12}.
Furthermore, high-quality DNA is relatively cheap to synthesize~\cite{hughes17}, all of which makes the chemical reaction network a promising development tool for biochemical applications.

The aim of this paper is to help further the \emph{reliability} of biochemical circuits.
In existing literature, the reliability of proposed circuit designs has been examined in one of two ways: simulation and experimentation.
References~\cite{jMagn97,jiang13,jGZWYZ17,beiki18} use simulation to analyze their circuits in various contexts, and references~\cite{arkin94,jQiaWin11a,garg18} include \emph{in vitro} experiments to verify their designs.
Although simulations and experiments demonstrate correctness under certain environmental assumptions and initial conditions, they cannot guarantee the \emph{absence} of failure.
Formally stating circuit requirements and rigorously proving their satisfaction in all circumstances satisfying certain conditions gives additional confidence in the design as well as insight into when failure is likely to occur.

We introduce a new biochemical circuit motif in the input/output chemical reaction network (I/O CRN) model originally introduced by Klinge, Lathrop, and Lutz~\cite{jack}.
An I/O CRN is an abstraction of the traditional CRN model~\cite{oFein79,oGuna03} making it possible for input signals to be provided externally over time.
These inputs can only be used catalytically which makes them \emph{read-only}.
Moreover, I/O CRNs offer a natural notion of robustness with respect to perturbations of the input signal, initial condition, rate constants, and measurement devices.
We use this notion to prove that our circuit designs operate correctly even in adversarial environments.

Our circuit design uses \emph{dual-rail} encoding of bits in which two species with opposite operational meaning are used to encode each value.
Each bit is designed so that the sum of these two species is constant, ensuring that if one has high concentration, the other is low.
Dual-rail representation is common in biochemical systems since both 0s and 1s are encoded by the \emph{presence} of molecules rather than their absence.
(Detecting the absence of a species is challenging since reactions are active only if their reactants are \emph{present}.
See~\cite{cDoty14,cEHKLLL14} for more details on the complexity of absence detection and for a proposed method for overcoming it.)
To ensure that only one of the dual-species is high at a time, we also include \emph{signal restoration} reactions for each encoded value.
These reactions are essential to proving that robustness is preserved under composition and causes the dual-species with majority concentration to consume the minority species.
For a thorough analysis of the behavior of these reactions, see~\cite{cKlin16}.

The key contributions of this work are: (1) we provide natural and rigorous requirements for what it means for I/O CRNs to simulate circuits; (2) we give an I/O CRN construction of a NAND gate and formally prove it satisfies its requirement even in the presence of perturbations to its input, initial state, rate constants, and measurements; (3) we prove that circuits can be modularly composed to robustly implement any combinatorial circuit; and (4) we prove that two commonly used sequential circuits for storing memory can be robustly implemented, namely the SR latch and the D latch.
Section~\ref{sec:prelim} reviews the I/O CRN model and the notion of robustly satisfying requirements;
Section~\ref{sec:robust_nand} provides an I/O CRN construction of a NAND gate with a formal proof that it is robust;
Section~\ref{sec:robust_circuits} contains our main theorem that all combinatorial circuits can be robustly implemented by I/O CRNs;
Section~\ref{sec:robust_memory} provides our I/O CRN constructions for the sequential memory components along with proofs that they are robust; and
Section~\ref{sec:discussion} closes with a discussion of the strengths and weaknesses of this method of implementing circuits.

\section{Preliminaries}\label{sec:prelim}
In this section, we review the definition of the input/output chemical reaction network (I/O CRN) and our notion of an I/O CRN robustly satisfying a requirement.
These were introduced by Klinge, Lathrop, and Lutz in 2016 and will soon appear in a detailed extension of~\cite{oKlLaLu15}.
For an in-depth overview, see~\cite{oKlin16}.

\subsection{Input/Output Chemical Reaction Networks}
\label{sub:iocrns}
We fix a countably infinite set $\bfm{S}=\{X_0,X_1,X_2\ldots\}$ of \emph{species}.
Intuitively, a species is an abstract type of molecule, and we denote them with capital Roman letters such as $X$, $Y$, and $Z$.
A \emph{reaction} over a finite set $\S\subseteq\bfm{S}$ of species is a triple $\rho=(\br,\bp,k)\in\svec{\S}\times\svec{\S}\times(0,\infty)$ such that $\br\ne\bp$.
The elements of a reaction $\rho=(\br,\bp,k)$ are called the \emph{reactant vector}, \emph{product vector} and \emph{rate constant}, respectively, and the \emph{net effect} of the reaction is the vector $\Delta\rho=\bp-\br$.
Given a reaction $\rho=(\br,\bp,k)$, we use $\brr=\br$, $\bpr=\bp$, and $\kr=k$ for the individual components of $\rho$.

We occasionally use the intuitive notation of chemistry to improve the readability of reactions.
For example, $A+B\goesto{k}2B+C$ defines the reaction $\rho=(\br,\bp,k)$ over the set $\S=\{A,B,C\}$  where $\br=(1,1,0)$ and $\bp=(0,2,1)$.
The net effect of the reaction is $\Delta\rho = (-1,1,1)$, meaning it consumes one $A$ and produces one $B$ and one $C$.
For convenience, we treat the vectors $\br$, $\bp$, and $\Delta\rho$ as functions from the set $\S$ into the natural numbers.
Thus, $\br(A)=1$, $\br(B)=1$, and $\br(C)=0$ for the reaction $\rho$ above.
We call a species $Y\in\S$ a \emph{reactant} of $\rho=(\br,\bp,k)$ if $\br(Y)>0$, a \emph{product} of $\rho$ if $\bp(Y)>0$, and a \emph{catalyst} of $\rho$ if $\br(Y)>0$ and $\Delta\rho(Y)=0$.
Note that a catalyst is simply a species that participates in a reaction but is unaffected by it.

An \emph{input/output chemical reaction network} (\emph{I/O CRN}) is a tuple $\N=(\U,\R,\S)$ where $\U,\S,\subseteq\bfm{S}$ are finite sets of species that satisfy $\U\cap \S=\emptyset$ and $\R$ is a finite set of reactions over $\U\cup\S$ such that $\Delta\rho(X)=0$ for each $\rho\in\R$ and $X\in\U$.
We call the elements of $\S$ \emph{state species} and the elements of $\U$ \emph{input species}.
Note that an I/O CRN $\N$ can only use its input species catalytically.
This ensures that input species are \emph{read-only} and cannot be modified by the operation of the I/O CRN.

Under \emph{deterministic mass action semantics} (also called \emph{mass action kinetics}), a \emph{state} of an I/O CRN $\N=(\U,\R,\S)$ is a vector $\bx\in\state{\S}$ that assigns to each $Y\in\S$ a real-valued \emph{concentration} $\bx(Y)$.
Similarly, an \emph{input state} is a vector $\bu\in\state{\U}$, and a \emph{global state} is a vector $(\bx,\bu)\in\state{\S\cup\U}$.

For a finite set $\W\subseteq\bfm{S}$, we define the $\W$-\emph{signal space} to be the set $C[\W]=C(\td,\state{\W})$ where $C(\mathcal{X},\mathcal{Y})$ is the set of all continuous functions from $\mathcal{X}$ to $\mathcal{Y}$.
A \emph{context} of an I/O CRN $\N=(\U,\R,\S)$ is a tuple $\bc=(\bu,\V,h)$ where $\bu\in C[\U]$, $\V\subseteq\S$, and $h:\state{\S\cup\U}\rightarrow\state{\V}$.
We call the components of the context $\bc=(\bu,\V,h)$ the \emph{input function}, the \emph{output species}, and the \emph{measurement function}, respectively.
The set of all contexts of an I/O CRN $\N$ is denoted $\mathcal{C}_\N$.
Intuitively, an I/O CRN can be regarded as a chemical machine that when placed in a context $\bc=(\bu,\V,h)$, transforms its input signal $\bu\in C[\U]$ into an observed output $\bv\in C[\V]$.
The inclusion of the measurement function $h$ in the definition of a context is to specify which species of the I/O CRN are being observed as well as encapsulate any errors introduced by the measurement equipment.
We also make use of the \emph{zero-error measurement function} $h_0$ defined by
\begin{equation}\label{eq:zero_error_measurement_function}
    h_0(\bx,\bu)(Y) = (\bx,\bu)(Y)
\end{equation}
for each global state $(\bx,\bu)\in\state{\S\cup\U}$ and for each state species $Y\in\S$.
Note that $h_0$ is a projection function and corresponds to a \emph{perfect} measurement device.

Given a global state $(\bx,\bu)\in\state{\S\cup\U}$ and a reaction $\rho\in\R$, the \emph{rate} of $\rho$ in $(\bx,\bu)$ is the real-value
\begin{equation}\label{eq:rate_definition}
    \rate{\bx,\bu}(\rho) = \kr\prod_{Y\in\S\cup\U}(\bx,\bu)(Y)^{\brr(Y)}.
\end{equation}
Thus, the rate of a reaction is proportional to each of its reactants.
For example, if $\rho=(\br,\bp,k)$ is the reaction defined by $A + 2 B \goesto{k} A + 3 C$ where $\U=\{A\}$ and $\S=\{B,C\}$, then its rate in state $(\bx,\bu)\in\state{\S\cup\U}$ is $k\bu(A)\bx(B)^2$.

For each species $Y\in\S$, the \emph{deterministic mass action function} for $Y$ is
\begin{equation}\label{eq:mass_action_function_definition}
    F_Y(\bx,\bu) = \sum_{\rho\in\R}\Delta\rho(Y)\cdot\rate{\bx,\bu}(\rho).
\end{equation}
Intuitively, the function $F_Y$ specifies the total rate of change imposed on $Y$ in the global state $(\bx,\bu)$.
In the context $\bc=(\bu,\V,h)$, the concentrations of all species in $\S$ of an I/O CRN evolve according to the system of ordinary differential equations (ODEs) defined by
\begin{equation}\label{eq:ode_vector_definition}
    \bx'(t) = F(\bx(t),\bu(t)),
\end{equation}
for all $t\in\td$ where $F(\bx,\bu)(Y)=F_Y(\bx,\bu)$ for each $Y\in\S$.
(Our occasional use of $\bx$ and $\bu$ as single states as well as concentration signals is intentional to reduce obfuscation.)

According to the standard theory of ODEs, if the input $\bu$ is real analytic, then the system~\eqref{eq:ode_vector_definition} along with an initial state $\xn\in\state{\S}$ has a unique solution $\bx(t)$ satisfying $\bx(0)=\xn$.
For this reason, we assume that all input signals are real analytic\footnote{
    All continuous signals produced by natural phenomena are real analytic, including all solutions to systems of polynomial differential equations.
    Therefore, placing this restriction on our input signals is not only necessary, it is a natural choice.
}.
See~\cite{oKraPar02} for a thorough introduction to real analytic functions.

Finally, we define the \emph{output signal} of an I/O CRN $\N=(\U,\R,\S)$ with initial state $\xn\in\state{\S}$ in context $\bc=(\bu,\V,h)$ to be
\begin{equation}\label{eq:output_function_definition}
    \N_{\xn,\bc}(t) = h(\bx(t)),
\end{equation}
for all $t\in\td$ where $\bx(t)$ is the unique solution to~\eqref{eq:ode_vector_definition} with initial state $\xn$.

We conclude by noting that I/O CRNs offer a natural means of modular design and composition.
Given two I/O CRNs $\N_1=(\U_1,\R_1,\S_1)$ and $\N_2=(\U_2,\R_2,\S_2)$, we define the \emph{join of} $\N_1$ and $\N_2$ to be the I/O CRN $\N_1\sqcup\N_2 = (\U,\R,\S)$ where $\U = \left(\U_1\cup\U_2\right)\setminus\left(\S_1\cup\S_2\right)$, $\R = \R_1\cup\R_2$, and $\S = \S_1\cup\S_2$.
If $\N_1$ and $\N_2$ have disjoint sets of state species, we say that $\N_1\sqcup\N_2$ is \emph{modular}.
Our combinatorial circuit architecture as well as our SR latch and D latch designs crucially depend on this natural modularity.

\subsection{Time-Dependent I/O CRNs}
\label{sub:time_dependent_crns}
In order to define robustness with respect to rate constants, we define a variation of the I/O CRN model that replaces the rate constants of reactions with non-negative functions of time.
For the purposes of this definition, we define a \emph{time-dependent reaction} over the set $\S$ to be tuple $\rho = (\br,\bp,\khat)$ where $\br,\bp\in\nats^{|\S\cup\U|}$ and $\khat:\td\rightarrow(0,\infty)$ is a real analytic function.
A \emph{time-dependent input/output chemical reaction network} (\emph{I/O tdCRN}) is a tuple $\N=(\U,\widehat{\R},\S)$ where $\U,\S\subseteq\bfm{S}$ are finite sets of species such that $\S\cap\U=\emptyset$ and $\widehat{\R}$ is a finite set of time-dependent reactions that only use species in $\U$ as catalysts.

The deterministic mass action semantics of an I/O tdCRN are the same as that of an I/O CRN except that the rate function of~\eqref{eq:rate_definition} changes to 
\begin{equation}\label{eq:rate_definition_td}
    \rate{\bx(t),\bu(t)}(\rho) = \khat(\rho)(t)\prod_{Y\in\S\cup\U}(\bx,\bu)(t)(Y)^{\brr(Y)},
\end{equation}
for all time $t\in\td$ in order to incorporate the time-dependent reactions.
Equations~\eqref{eq:mass_action_function_definition}-\eqref{eq:output_function_definition} also change using this new rate equation and become
\begin{align}
    F_Y(\bx(t),\bu(t)) &= \sum_{\rho\in\R}\Delta\rho(Y)\cdot\rate{\bx(t),\bu(t)}(\rho)\label{eq:mass_action_function_td}\\
    \bx'(t) &= F(\bx(t),\bu(t))\label{eq:ode_vector_definition_td}\\
    \N_{\xn,\bc}(t) &= h(\bx(t)),
\end{align}
respectively.

For an I/O CRN $\N=(\U,\R,\S)$ and constant $\delta>0$, we say that an I/O tdCRN $\Nhat=(\U,\widehat{\R},\S)$ is $\delta$-\emph{close} to $\N$ if each $\hat{\rho}\in\widehat{\R}$ is the time-dependent equivalent of $\rho\in\R$ and satisfies $|\kr-\khat(\hat{\rho})(t)|\le\delta$ for all $t\in\td$.

\subsection{Robustness}\label{sub:robustness}
A \emph{requirement} of an I/O CRN $\N=(\U,\R,\S)$ is an ordered-pair $\Phi=(\alpha,\phi)$ consisting of the two Boolean predicates $\alpha:\mathcal{C}_\N\rightarrow\{\true,\false\}$ and $\phi:C[\U]\times C[\V]\rightarrow\{\true,\false\}$, called the \emph{context assumption} and the \emph{I/O requirement}, respectively.
We say that an I/O CRN $\N=(\U,\R,\S)$ \emph{satisfies} the requirement $\Phi=(\alpha,\phi)$, and we write $\N\models\Phi$, if there exists an initial state $\xn\in\state{\S}$ such that for all $\bc\in\mathcal{C}_\N$
\begin{equation}
    \alpha(\bc) \implies \phi(\bu,\N_{\xn,\bc}).
\end{equation}

In order to capture the notion of approximately satisfying a requirement, we use the \emph{supremum norm} $\norm{f}=\sup_{t\in\td}|\bw(t)|$ for all $\bw\in C[\W]$ where $|\bw(t)|=\sqrt{\sum_{Y\in\W}\bw(t)(Y)^2}$ is the Euclidean distance function in $\reals^{|\W|}$.
For $\bw\in C[\W]$ and $\epsilon>0$, we define the \emph{closed ball of radius} $\epsilon$ \emph{around} $\bw$ to be the set $B_\epsilon(\bw) = \{\widehat{\bw}\mid\norm{\bw-\widehat{\bw}}\le\epsilon\}$.
If $\widehat{\bw}\in B_\epsilon(\bw)$, we say that $\widehat{\bw}$ is $\epsilon$-\emph{close} to $\bw$.

An I/O CRN $\N=(\U,\R,\S)$ $\epsilon$-\emph{satisfies} a requirement $\Phi=(\alpha,\phi)$, and we write $\N\models_\epsilon\Phi$, if there exists an initial state $\xn\in\state{\S}$ such that
\begin{equation}
    \alpha(\bu,\V,h)\implies\exists\bv\in B_\epsilon(\N_{\xn,\bc})\;[\phi(\bu,\bv)].
\end{equation}

Given a context $\bc=(\bu,\V,h)$ and real numbers $\delta_1,\delta_2>0$, we say that $\bchat=(\buhat,\V,\hhat)$ is $(\delta_1,\delta_2)$-\emph{close} to $\bc$ if $\norm{\bu-\buhat}\le\delta_1$ and $\norm{h-\hhat}\le\delta_2$.
Given states $\bx,\bxhat\in\state{\S}$ and $\delta>0$, we say that $\bxhat$ is $\delta$-\emph{close} to $\bx$ if $|\bx-\bxhat|\le\delta$.

Finally we state what it means for an I/O CRN to robustly satisfy a requirement.
Given $\N=(\U,\R,\S)$, $\Phi=(\alpha,\phi)$, $\epsilon>0$, and $\bfm{\delta}=(\delta_1,\delta_2,\delta_3,\delta_4)$ such that $\delta_1,\delta_2,\delta_3,\delta_4>0$,
we say that $\N$ $\bfm{\delta}$-\emph{robustly} $\epsilon$-\emph{satisfies} $\Phi$, and we write $\N\models_\epsilon^{\bfm{\delta}}\Phi$,
if there exists an initial state $\xn\in\state{\S}$ such that for all contexts $\bc=(\bu,\V,h)$ satisfying $\alpha(\bc)$,
for each context $\bchat=(\buhat,\V,\hhat)$ $(\delta_1,\delta_2)$-close to $\bc$,
for each state $\xnhat\in\state{\S}$ $\delta_3$-close to $\xn$,
and for each I/O tdCRN $\widehat{\N}$ $\delta_4$-close to $\N$, there exists a concentration signal $\bv\in C[\V]$ that is $\epsilon$-close to the output signal $\widehat{\N}_{\xnhat,\bchat}$ that satisfies $\phi(\bu,\bv)$.

We conclude this section with a note on modularly joining I/O CRNs.
If $\N_1$ and $\N_2$ are two I/O CRNs satisfying $\N_1\models_{\epsilon_1}^{\bfm{\delta}_1}\Phi_1$ and $\N_2\models_{\epsilon_2}^{\bfm{\delta}_2}\Phi_2$, respectively, and $\N=\N_1\sqcup\N_2$ is a modular join of $\N_1$ and $\N_2$, then the individual subcomponents of $\N$ still satisfy the requirements $\Phi_1$ and $\Phi_2$.
However, if $\N_1$ and $\N_2$ share state species, it is possible for them to interfere with each other, and they may no longer satisfy $\Phi_1$ and $\Phi_2$ after the join.
We utilize this modular composition extensively throughout the paper.

\section{A Robust NAND Gate}
\label{sec:robust_nand}
In this section, we prove that a two-input NAND gate can be robustly implemented by an I/O CRN.
First, we formally specify the requirement, then we give our I/O CRN implementation, and finally we prove the construction robustly satisfies the requirement.

Since the inputs and output of the NAND gate are implicit parameters to the requirement, we start by specifying them.
Given $X_1,X_2\in\bfm{S}$, we define the set of input species to be $\U = \{X_1,X_2, \Xbar_1, \Xbar_2\}\subseteq\bfm{S}$.
The species $X_1$ and $X_2$ represent the two inputs of the NAND gate, and $\Xbar_1$ and $\Xbar_2$ are their \emph{duals}.
A \emph{dual} of a species represents its Boolean complement; thus, if the concentration of $X_1$ is $b\in\{0,1\}$, the concentration of $\Xbar_1$ is $1-b$.
We also use this dual-rail convention for the output, and let $V=\{Y,\Ybar\}\subseteq\bfm{S}$ be the set of output species given $Y\in\bfm{S}$.

Given a positive real number $\tau$, called the \emph{propagation delay}, we define the NAND gate requirement $\Phi_\text{NAND}(\tau) = (\alpha,\phi)$ where $\alpha$ is defined by
\begin{equation}
    \alpha(\bu,\V,h) \equiv \left[\V=\{Y,\Ybar\}\text{ and }h=h_0\right],
\end{equation}
where $h_0$ from equation~\eqref{eq:zero_error_measurement_function} is the zero-error measurement function.
Requiring that $h=h_0$ simply requires it to faithfully measure the output species concentrations.
Errors will eventually be introduced into $h$ when we show that $\Phi_\text{NAND}(\tau)$ is \emph{robustly} satisfied.

Before we specify the I/O requirement of $\Phi_\text{NAND}(\tau)$, we first define some useful notation.
Let $\bfm{I}(\tau)$ be the set of all closed intervals at least length $\tau$, defined by
\begin{equation}
    \bfm{I}(\tau)=\{I=[t_1,t_2]\subseteq\td\mid t_2-t_1\ge \tau\}.
\end{equation}
Since the I/O requirement $\phi$ is a predicate that takes parameters $\bu\in C[\U]$ and $\bv\in C[\V]$, we use $\bu$ and $\bv$ as implicit parameters in the following definitions.
Given an interval $I\in\bfm{I}(\tau)$, a species $W\in\U\cup\V$, and a bit $a\in\{0,1\}$, we define
\[
    \dbracket{W=a}_I
        \equiv\;
        \begin{cases}
            (\forall t\in I)\big[\;\bu(t)(W) = a = 1-\bu(t)(\Wbar)\;\big], &\text{ if }W\in\U\\
            (\forall t\in I)\big[\;\bv(t)(W) = a = 1-\bv(t)(\Wbar)\;\big], &\text{ if }W\in\V
        \end{cases}.
\]
Note that $\dbracket{W=a}_I$ simply says that the species $W$ and its dual encode the values $a$ and $1-a$ for all $t\in I$.
To help with our definition of $\phi$, we also define the predicates
\[
    \phi_{11}(I)
        \equiv\dbracket{X_1=1\land X_2=1}_I,
    \qquad
    \phi_0(I)
        \equiv\dbracket{X_1=0\lor X_2=0}_I,
\]
for all $I\in\bfm{I}(\tau)$.
The predicate $\phi_{11}(I)$ says that $X_1$ and $X_2$ both encode the value $1$ in $I$ and $\phi_0(I)$ says that at least one of $X_1$ and $X_2$ must encode $0$ in $I$.
Similarly, for $a\in\{0,1\}$ we define the Boolean predicate
\[
    \psi_a(I)\equiv\dbracket{Y=a}_{[t_1+\tau,t_2]},
\]
for all $I=[t_1,t_2]\in\bfm{I}(\tau)$, which says that $Y$ encodes $a$ for all but the first $\tau$ time of the interval $I$.

We now have sufficient terminology to define the I/O requirement $\phi$ to be
\begin{equation}
    \phi(\bu,\bv)
        \equiv\big(\forall I\in\bfm{I}(\tau)\big)
            \big[
                \left(\phi_{11}(I)\rightarrow\psi_0(I)\right)
                \land
                \left(\phi_0(I)\rightarrow\psi_1(I)\right)
            \big]
\end{equation}
for all $\bu\in C[\U]$ and $\bv\in C[\V]$.
Intuitively, $\phi$ says that if $X_1$ and $X_2$ are both 1, then $Y$ must converge to 0 in at most $\tau$ time and must remain there as long as both inputs stay 1.
Similarly, if either input is 0, then the output must converge to 1 in at most $\tau$ time and remain there while the 0 persists.
This is visualized in Figure~\ref{fig:nand_requirement}.

\begin{figure}
    \centering
    \def\svgwidth{\textwidth}
    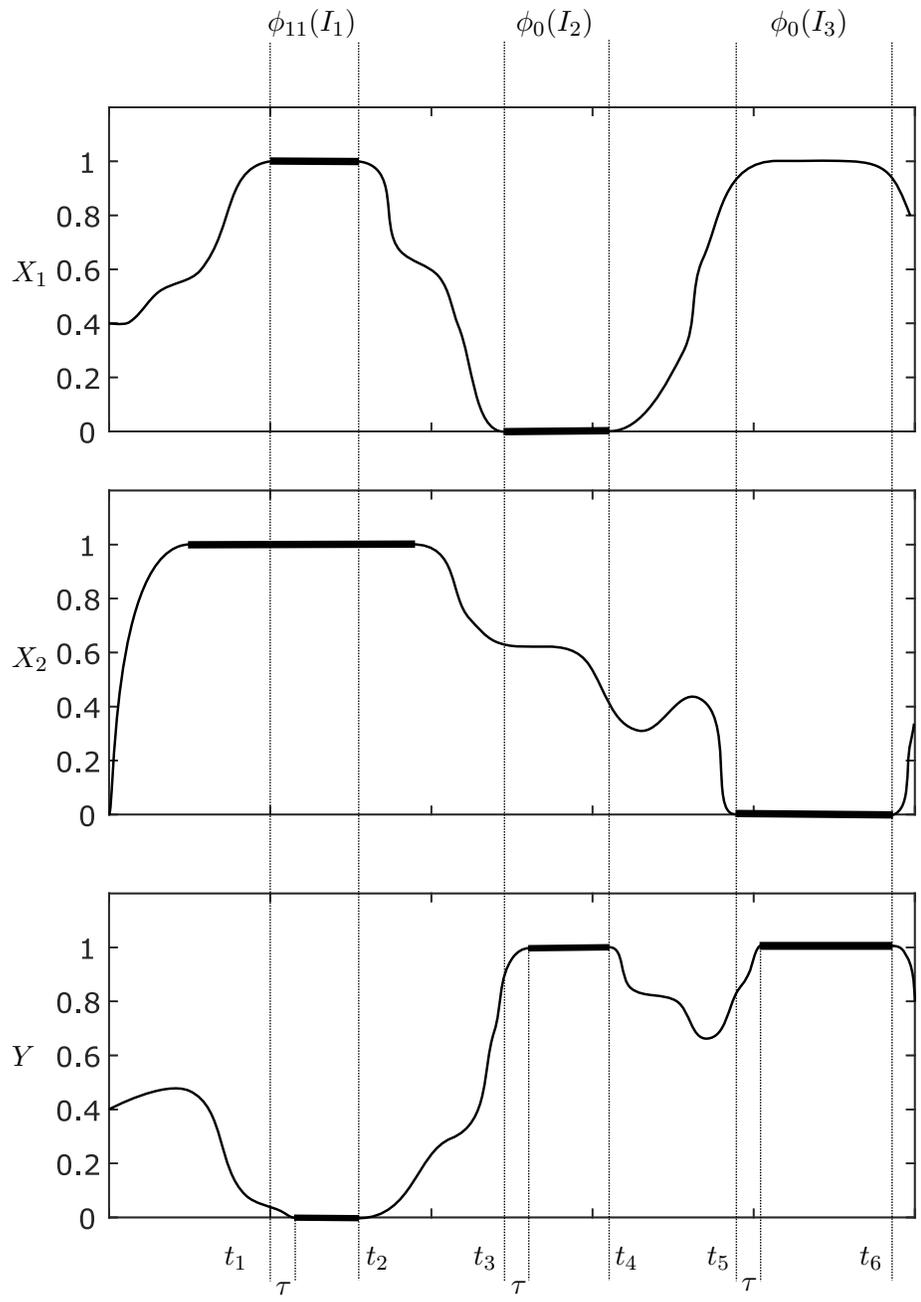
    \caption{Visualization of the NAND gate requirement
            \label{fig:nand_requirement}}
\end{figure}

We now specify our I/O CRN that robustly simulates a NAND gate.

\begin{construction}\label{const:robust_nand}
    Given three species $X_1,X_2,Y\in\bfm{S}$, a vector of strictly positive real numbers $\bfm{\delta}=(\delta_1,\delta_2,\delta_3,\delta_4)$, and $\tau>0$, define the I/O CRN $\text{NAND}_{\bfm{\delta},\tau}(X_1,X_2,Y) = (\U,\R,\S)$, where $\U = \{X_1,X_2,\Xbar_1,\Xbar_2\}$, $\S = \{Y, \Ybar\}$, $\R$ consists of the reactions
    \begin{align}
        X_1 + X_2 + Y &\goesto{k} X_1 + X_2 + \Ybar\label{reaction:set_one}\\
        \Xbar_1 + \Ybar &\goesto{k} \Xbar_1 + Y\label{reaction:set_zero_1}\\
        \Xbar_2 + \Ybar &\goesto{k} \Xbar_2 + Y\label{reaction:set_zero_2}\\
        2 Y + \Ybar &\goesto{3k} 3Y\label{reaction:approx_maj_1}\\
        2 \Ybar + Y &\goesto{3k} 3\Ybar\label{reaction:approx_maj_2},
    \end{align}
    and where $k = 100\delta_4 + \frac{13}{\tau}$.
\end{construction}

In the above construction, reaction~\eqref{reaction:set_one} biases the output toward $\Ybar$ when the inputs $X_1$ and $X_2$ are both present, reactions~\eqref{reaction:set_zero_1}-\eqref{reaction:set_zero_2} bias the output toward $Y$ in the presence of $\Xbar_1$ or $\Xbar_2$ (\ie in the absence of $X_1$ or $X_2$), and reactions~\eqref{reaction:approx_maj_1}-\eqref{reaction:approx_maj_2} give extra bias to the output species with majority concentration.
The latter two reactions are essential for the I/O CRN to produce an output signal that is as clean as its input and was studied extensively in \cite{cKlin16}.
The construction also preserves the total concentration of $Y$ and $\Ybar$ so that their sum is always constant.

We now state the main theorem of this section.

\begin{theorem}\label{theorem:robust_nand}
    If $\bfm{\delta}=(\delta_1,\delta_2,\delta_3,\delta_4)\in(0,\infty)^4$ and $\tau > 0$ are constants satisfying $\delta_2+\delta_3 < \delta_1 < \frac{1}{25}$ and $\delta_2+\delta_3 < \frac{1}{100}$, then $\text{NAND}_{\bfm{\delta},\tau}(X_1,X_2,Y)\models_{\delta_1}^{\bfm{\delta}} \Phi_\text{NAND}(\tau)$.
\end{theorem}

The remainder of this section is devoted to proving this theorem.
Since the proof requires examining an arbitrary perturbation of a variety of parameters, we begin the proof by fixing these perturbations.

Assume the hypothesis with $\N=(\U,\R,\S)=\text{NAND}_{\bfm{\delta},\tau}(X_1,X_2,Y)$.
We fix initial state $\xn\in\td^\S$ defined by $\xn(Y)=1$ and $\xn(\Ybar)=0$.
(Note that any choice satisfying $\xn(Y)+\xn(\Ybar)=1$ suffices for our argument.)
Let $\bc=(\bu,\V,h)$ be a context that satisfies the context assumption $\alpha(\bc)$.
Let $\bchat=(\buhat,\V,\hhat)$ be $(\delta_1,\delta_2)$-close to $\bc$, let $\xnhat$ be $\delta_3$-close to $\xn$, and let $\Nhat$ be $\delta_4$-close to $\N$.
It now suffices to show that the output function $\Nhat_{\bchat,\xnhat}$ is $\delta_1$-close to a signal $\bv\in C[\V]$ satisfying $\phi(\bu,\bv)$ of $\Phi_\text{NAND}$.
Let $\bxhat\in C[\S]$ as the unique solution generated by $\Nhat$ in context $\bchat$ on the initial state $\xnhat$.
For convenience, we write $y(t)$ and $\ybar(t)$ to denote $\bxhat(t)(Y)$ and $\bxhat(t)(\Ybar)$, respectively.

Using the reactions from Construction~\ref{const:robust_nand} along with the definition of the deterministic mass action system for an I/O tdCRN from equations~\eqref{eq:rate_definition_td}-\eqref{eq:ode_vector_definition_td}, we observe that the ODEs for $y(t)$ and $\ybar(t)$ are
\begin{align}
    \D{y} &= 3\khat_1 y^2\ybar - 3\khat_2 y \ybar^2 - \khat_3x_1x_2y + \khat_4\xbar_1\ybar + \khat_5\xbar_2\ybar,\label{eq:ode_y}\\
    \D{\ybar} &= -\D{y},\label{eq:ode_ybar}
\end{align}
where $\khat_1$, $\khat_2$, $\khat_3$, $\khat_4$, and $\khat_5$ are all time-varying $\delta_4$-perturbations of the rate constant $k$ and $x_1(t)$, $x_2(t)$, $\xbar_1(t)$, and $\xbar_2(t)$ are the four components of the $\delta_1$-perturbed input signal $\buhat(t)$.

Equation~\eqref{eq:ode_ybar} immediately implies that the total concentration of $Y$ and $\Ybar$ is constant, \ie, that $p=y(t)+\ybar(t)$ for all $t\in\td$ where $p=\xnhat(Y)+\xnhat(\Ybar)$.
It is also useful to note that $|p-1|<\delta_3$ since $\xnhat$ is a $\delta_3$-perturbation of $\xn$ which satisfies $\xn(Y)+\xn(\Ybar)=1$.

The I/O requirement $\phi(\bu,\bv)$ is the conjunction of two statements, and we prove each statement holds individually in Lemmas~\ref{lemma:case_11} and~\ref{lemma:case_0}.
Before proving these lemmas, we show that the solution $\bxhat(t)$ is bounded by the solution of much simpler systems of ODEs, and the analyses of these simpler ODEs are given in Lemmas~\ref{lemma:x_goesto_3/5} and \ref{lemma:x_signal_restoration}.
For convenience, we define the constant $d=\frac{\delta_4}{k}$.
\begin{lemma}\label{lemma:x_goesto_3/5}
    If $x(t)$ is the solution to the IVP defined by $x(0)=0$ and
    \begin{equation}\label{eq:simple_x_ode}
        \D{x} = k\left(-a  + b(p-x) - cx\right),
    \end{equation}
    where $a = \frac{p^3}{18} \left((3+d)^{3/2}+ 9d\right)$, $b = (1-d)(1-\delta_1)^2$, and $c = 2\delta_1(1+d)$, then $x(\frac{\tau}{2})>\frac{3}{5}$.
\end{lemma}
\begin{proof}
    The single variable ODE~\eqref{eq:simple_x_ode} can be solved by separation of variables and integrating which yields
    \[
        x(t) = \frac{b p-a}{b+c}\left(1-e^{-k(b+c)t}\right).
    \]
    Using the facts that $\delta_1<\frac{1}{25}$, $d<\frac{1}{100}$, $\delta_3<\frac{1}{100}$, $|p-1|<\delta_3$ and $k>\frac{13}{\tau}$, it is easy to verify via substitution that $x\left(\frac{\tau}{2}\right) > \frac{3}{5}$.
\end{proof}

\begin{lemma}\label{lemma:x_signal_restoration}
    If $x(t)$ is the solution to the IVP defined by $x(0)=\frac{3}{5}$ and
    \begin{equation}\label{eq:simple_x_ode_2}
        \D{x} = ax^2(p-x) - bx(p-x)^2 -cx,
    \end{equation}
    where $a=3k(1-d)$, $b=3k(1+d)$, and $c=2k\delta_1(1+d)$,
    then $x(t)>p-\gamma$ for all $t\ge\frac{\tau}{2}$ where $\gamma=\delta_1-\delta_2-\delta_3$.
\end{lemma}
\begin{proof}
    The ODE~\eqref{eq:simple_x_ode_2} has been studied extensively and is sometimes referred to as a \emph{signal restoration algorithm}.
    According to two theorems proved in~\cite{cKlin16}, if the inequalities
    \begin{align}
        c <\frac{p^2a^2}{4(a+b)}\\
        x(0) > E_1,
    \end{align}
    hold where $E_1=p\left(\frac{b}{a+b}\right)+A$ such that $A=\frac{p}{2}\left(\frac{a}{a+b}\right)\left(1-\sqrt{1-c^\ast}\right)$ and $c^\ast=\frac{4c(a+b)}{p^2a^2}$, then $x(t)$ exponentially quickly converges to the value $E_2=p-A$.
    Using the facts that $d<\frac{1}{100}$, $\delta_1<\frac{1}{25}$, $\delta_3<\frac{1}{100}$ and $x(0)=\frac{3}{5}$, it is easy to verify that both of the above inequalities hold.

    Corollary 4.5 of~\cite{cKlin16} shows that under these conditions $x(t)$ will converge to the quantity $p-\gamma$ and remain above it indefinitely in at most time
    \[
        T = \frac{a+b}{abp^2(1-c^\ast)}\log u,
    \]
    where $u=\frac{(p-\gamma-E_1)(E_2-\frac{3}{5})}{(\frac{3}{5}-E_1)(E_2-p+\gamma)}$.
    Using the bounds of $d$, $\delta_1$, and $\delta_3$ and the fact that $k>\frac{13}{\tau}$, it is easy to verify that $T\le\frac{\tau}{2}$.
    Thus, $x(t) > p-\gamma$ for $t\ge\frac{\tau}{2}$.
\end{proof}

\begin{lemma}\label{lemma:case_11}
    If $I\in\bfm{I}(\tau)$ such that $\phi_{11}(I)$ holds, then $\psi_0(I)$ holds.
\end{lemma}
\begin{proof}
    Assume the hypothesis for $I=[t_1,t_2]\in\bfm{I}(\tau)$.
    To show that $\psi_0(I)$ holds, we need to show that $1-\delta_2<\ybar(t)<1+\delta_2$ and $y(t) < \delta_2$ holds for all $t\in [t_1+\tau,t_2]$.
    Since $y(t)+\ybar(t) = p$, it suffices to show that $\ybar(t)>p-\gamma$ where $\gamma = \delta_1-\delta_2-\delta_3$ for all $t\in[t_1+\tau,t_2]$.
    We will show this by bounding the ODE of $\Ybar$ from equation~\eqref{eq:ode_ybar}.

    Since the perturbed rate constants are within $\delta_4$ of $k$, we know that
    \[
        \D{\ybar} \ge 3(k-\delta_4)\ybar^2y - 3(k+\delta_4)\ybar y^2 + (k-\delta_4)\khat_3x_1x_2y
            - (k+\delta_4)\khat_4\xbar_1\ybar - (k+\delta_4)\xbar_2\ybar.
    \]
    Thus if we let $d=\frac{\delta_4}{k}$, we can rewrite this equation as
    \begin{equation}\label{eq:ode_upper_bound_1}
        \D{\ybar} \ge k\big[3(1-d)\ybar^2y 
            - 3(1+d)\ybar y^2 + (1-d)x_1x_2y
            - (1+d)(\xbar_1+\xbar_2)\ybar\big].
    \end{equation}
    It is also not difficult to show that the expression $3(1-d)\ybar^2y - 3(1+d)\ybar y^2$ is minimized by letting $\ybar=\frac{p}{6}\left(d+3-\sqrt{d^2+3}\right)$.
    By substituting this into the expression, we obtain
    \begin{align*}
        3(1-d)\ybar^2y - 3(1+d)\ybar y^2
            &\ge -\frac{p^3}{18} \left(3 \sqrt{d^2+3}+d \left(d \left(\sqrt{d^2+3}-d\right)+9\right)\right)\\
            &\ge -\frac{p^3}{18} \left((3+d)^{3/2} + 9d\right).
    \end{align*}
    After substituting this into~\eqref{eq:ode_upper_bound_1} we obtain the bound
    \[
        \D{\ybar} \ge k\left[-\frac{p^3}{18} \left((3+d)^{3/2} + 9d\right) + (1-d)x_1x_2y - (1+d)(\xbar_1+\xbar_2)\ybar\right].
    \]

    Since $\phi_{11}(I)$ holds, we know that within the interval $I$ that $x_1$, $x_2$, $\xbar_1$, and $\xbar_2$ are encoding 1, 1, 0, and 0, respectively.
    However, these are only $\delta_1$-approximating these because of the input perturbation.
    Thus, for all $t\in I$ we have
    \[
        \D{\ybar} \ge k\left[-a  + b(p-\ybar) - c\ybar\right],
    \]
    where $a=\frac{p^3}{18} \left((3+d)^{3/2}+ 9d\right)$, $b=(1-d)(1-\delta_1)^2$, and $c=2\delta_1(1+d)$.
    By Lemma~\ref{lemma:x_goesto_3/5}, we know $\ybar(t_1+\frac{\tau}{2})\ge\frac{3}{5}$.

    To bound the behavior of $\Ybar$ after time $t_1+\frac{\tau}{2}$, we take another look at~\eqref{eq:ode_upper_bound_1} and see that
    \begin{align*}
        \D{\ybar}
            &\ge k\left[3(1-d)\ybar^2y - 3(1+d)\ybar y^2 - 2\delta_1(1+d)\ybar\right]\\
            &\ge a\ybar^2(p-\ybar) - b\ybar(p-\ybar)^2 -c\ybar,
    \end{align*}
    where $a=3k(1-d)$, $b=3k(1+d)$, and $c=2k\delta_1(1+d)$.
    By Lemma~\ref{lemma:x_signal_restoration}, we see that $\ybar(t)>p-\gamma$ for all $t\in[t_1+\tau,t_2]$ which also means that $y(t)<\gamma$ during that interval since $y(t)+\ybar(t)=p$.

    Finally, since $p>1-\delta_3$, $\gamma = \delta_1-\delta_2-\delta_3$, and the measurement function can only introduce $\delta_2$ amount of error, $\Nhat_{\xnhat,\bchat}(t)(\Ybar)>1-\delta_1$ and $\Nhat_{\xnhat,\bchat}(t)(Y)<\delta_1$.
    Therefore $\Nhat_{\xnhat,\bchat}$ is $\delta_1$-close to encoding an output of $Y=0$ and $\Ybar=1$ in the interval $[t_1+\tau,t_2]$.
\end{proof}

\begin{lemma}\label{lemma:case_0}
    If $I\in\bfm{I}(\tau)$ such that $\phi_0(I)$ holds, then $\psi_1(I)$ holds.
\end{lemma}
\begin{proof}
    During an interval $I=[t_1,t_2]$ satisfying $\phi_0(I)$, it is easy to show by a similar argument to Lemma~\ref{lemma:case_11} that the inequalities
    \[
        \D{y}\ge
            k\left[-\frac{p^3}{18} \left((3+d)^{3/2} + 9d\right) + (1-d)(1-\delta_1)(p-y) - 2(1+d)\delta_1\ybar\right]
    \]
    and
    \[
        \D{y}\ge
            k\left[3(1-d)y^2(p-y) - 3(1+d)y (p-y)^2 - 2\delta_1(1+d)y\right]
    \]
    hold for all $t\in I$.
    Thus by Lemmas~\ref{lemma:x_goesto_3/5} and \ref{lemma:x_signal_restoration}, we see that $y(t)>p-\gamma$ and $\ybar(t)<\gamma$ for all $t\in[t_1+\tau,t_2]$, and thus $\phi_1(I)$ holds.
\end{proof}

\section{Robust Combinatorial Circuits}
\label{sec:robust_circuits}
In this section, we state and prove our main theorem, namely, that every combinatorial circuit can be implemented with an I/O CRN.
For each combinatorial circuit, we define its requirement, give an I/O CRN construction for it, and prove the construction robustly satisfies its corresponding requirement.

Given positive integers $n,m>0$, we define an $n$-\emph{input} $m$-\emph{output combinatorial circuit} $C_{n,m}$ to be a directed acyclic graph where each node is a two-input one-output NAND gate.
The circuit $C_{n,m}$ has $n$ incoming edges called \emph{inputs} and $m$ outgoing edges called \emph{outputs}.
The \emph{depth} of a circuit $C_{n,m}$ is the longest path from an input to an output.
Each circuit $C_{n,m}$ can be regarded as a function $C_{n,m}:\{0,1\}^n\rightarrow\{0,1\}^m$ defined in the obvious way by computing the values of the outputs by propagating the input values through each of the NAND gates of the circuit.
Since NAND gates are universal for combinatorial circuits, this definition includes all possible functions for this class.
Furthermore, our dual-rail scheme gives access to the negation of each signal without any additional gates.
This substantially reduces the size of many circuits.

For a circuit $C_{n,m}$, we define the set of input species to be
\[
   \U = \{X_i,\Xbar_i\mid 0\le i < n\}\subseteq\bfm{S},
\]
and define the requirement $\Phi(C_{n,m},\tau)=(\phi,\alpha)$ where $\alpha$ is defined by
\begin{equation}
    \alpha(\bu,\V,h)\equiv\left[\V=\{Y_i,\Ybar_i\mid 0\le i < m\}\text{ and }h=h_0\right].
\end{equation}
To state the I/O requirement $\phi$, we need a bit more terminology.
For a string $w\in\{0,1\}^n$ and input $\bu\in C[\U]$, we use the notation $\bu(t)=w$ to denote that $\bu(t)(X_i)=w[i]$ and $\bu(t)(\Xbar_i)=1-w[i]$ for each $0\le i < n$.
We also define the predicates
\[
    \phi_w(I)\equiv(\forall t\in I)\big[\bu(t) = w\big],
    \qquad
    \psi_w(I)\equiv(\forall t\in [t_1+\tau,t_2])\big[\bv(t) = w\big],
\]
for all $I=[t_1,t_2]\in\bfm{I}(\tau)$.
The I/O requirement $\phi$ can then be defined by
\begin{equation}
    \phi(\bu,\bv)
        \equiv
        \big(\forall I\in\bfm{I}(\tau)\big)
        \big(\forall w\in\{0,1\}^n\big)
        \big[\phi_w(I)\rightarrow\psi_{C_{n,m}(w)}(I)\big].
\end{equation}
Thus, $\Phi(C_{n,m},\tau)$ simply requires that an I/O CRN generates the output $C_{n,m}(w)$ within $\tau$ time whenever the inputs encode $w\in\{0,1\}^n$.

We now give the I/O CRN construction for an arbitrary combinatorial circuit.
\begin{construction}\label{const:robust_circuit}
    Given a combinatorial circuit $C_{n,m}$ with $G$ gates and depth $d$ along with constants $\bfm{\delta}=(\delta_1,\delta_2,\delta_3,\delta_4)$, and $\tau>0$, define the CRN $\N(C_{n,m},\bfm{\delta},\tau)$ by joining $G$ copies of the I/O CRN $\text{NAND}_{\bfm{\delta},\frac{\tau}{d}}$ from Construction~\ref{const:robust_nand} according to the circuit $C_{n,m}$.
\end{construction}

As an example, consider a two-input one-output exclusive or (XOR) circuit.
Since negations are free in our motif, the XOR circuit can be constructed using three NAND gates, depicted in Figure~\ref{fig:54}.
\begin{figure}
    \centering
    \def\svgwidth{0.5\textwidth}
\begingroup%
  \makeatletter%
  \providecommand\color[2][]{%
    \errmessage{(Inkscape) Color is used for the text in Inkscape, but the package 'color.sty' is not loaded}%
    \renewcommand\color[2][]{}%
  }%
  \providecommand\transparent[1]{%
    \errmessage{(Inkscape) Transparency is used (non-zero) for the text in Inkscape, but the package 'transparent.sty' is not loaded}%
    \renewcommand\transparent[1]{}%
  }%
  \providecommand\rotatebox[2]{#2}%
  \newcommand*\fsize{\dimexpr\f@size pt\relax}%
  \newcommand*\lineheight[1]{\fontsize{\fsize}{#1\fsize}\selectfont}%
  \ifx\svgwidth\undefined%
    \setlength{\unitlength}{204.36215702bp}%
    \ifx\svgscale\undefined%
      \relax%
    \else%
      \setlength{\unitlength}{\unitlength * \real{\svgscale}}%
    \fi%
  \else%
    \setlength{\unitlength}{\svgwidth}%
  \fi%
  \global\let\svgwidth\undefined%
  \global\let\svgscale\undefined%
  \makeatother%
  \begin{picture}(1,0.77521465)%
    \lineheight{1}%
    \setlength\tabcolsep{0pt}%
    \put(0,0){\includegraphics[width=\unitlength,page=1]{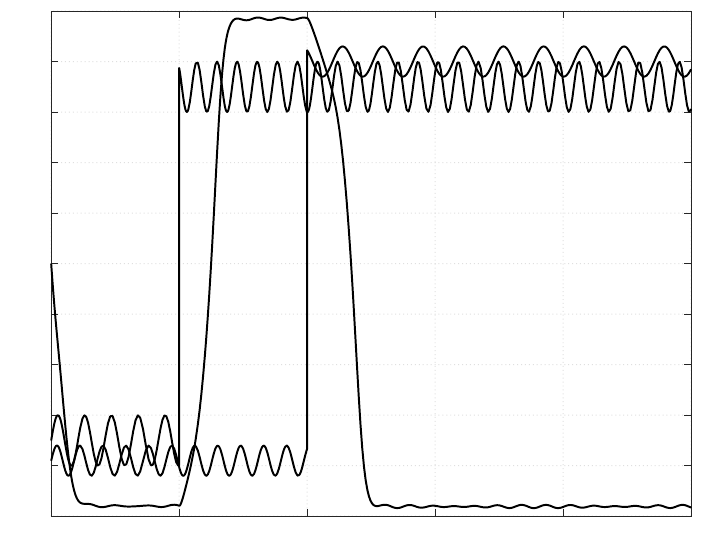}}%
    \put(0.165,0.23609322){\color[rgb]{0,0,0}\makebox(0,0)[lt]{\smash{\begin{tabular}[t]{l}$X_1$\end{tabular}}}}%
    \put(0.35,0.237154){\color[rgb]{0,0,0}\makebox(0,0)[lt]{\smash{\begin{tabular}[t]{l}$X_2$\end{tabular}}}}%
    \put(0.57816004,0.08433618){\color[rgb]{0,0,0}\makebox(0,0)[lt]{\smash{\begin{tabular}[t]{l}$Y$\end{tabular}}}}%
    \put(0,0){\includegraphics[width=\unitlength,page=2]{Figure54New.pdf}}%
  \end{picture}%
\endgroup%

    \raisebox{0.6in}[0pt][0pt]{
        \begin{tikzpicture}[align=center]
            \node[nand gate US, draw, rotate=0, logic gate inputs=in] at (0,1) (n1) {};
            \node[nand gate US, draw, rotate=0, logic gate inputs=ni] at (0,0) (n2) {};
            \node[nand gate US, draw, rotate=0, logic gate inputs=nn] at ($(n2) + (1,0.5)$) (n3) {};
            \node (x1) at ($(n1.input 1) + (-1,0)$) {$X_1$};
            \node (x2) at ($(n2.input 2) + (-1,0)$) {$X_2$};

            \draw (x1) -- (n1.input 1);
            \draw (x2) -- (n2.input 2);
            \draw (x1) -- ($(n1.input 1) + (-0.5,0)$) node[branch] {} |- (n2.input 1);
            \draw (x2) -- ($(n2.input 2) + (-0.15,0)$) node[branch] {} |- (n1.input 2);
            \draw (n1.output) -- ($(n1.output) + (0.1,0)$) |- (n3.input 1);
            \draw (n2.output) -- ($(n2.output) + (0.1,0)$) |- (n3.input 2);
            \draw (n3.output) -- node[above] {$Y$} ($(n3.output)+(0.5,0)$);
        \end{tikzpicture}
    }
    \caption{XOR circuit with sinusoidal noise\label{fig:54}}
\end{figure}
According to Construction~\ref{const:robust_circuit}, the I/O CRN defined by this circuit is
\[
   \N = \N(C_{n,m},\bfm{\delta},\tau)=\N_1\sqcup\N_2\sqcup\N_3,
\]
where
\begin{align*}
    \N_1&=\text{NAND}_{\bfm{\delta},\frac{\tau}{2}}(\Xbar_1,X_2,Z_1),\\
    \N_2&=\text{NAND}_{\bfm{\delta},\frac{\tau}{2}}(X_1,\Xbar_2,Z_2),\text{ and}\\
    \N_3&=\text{NAND}_{\bfm{\delta},\frac{\tau}{2}}(Z_1,Z_2,Y).
\end{align*}
For convenience, we assume that the dual of $\Xbar$ is $X$ so that negations are handled intuitively.
The unlabeled intermediate wires correspond to the state species $Z_1$ and $Z_2$ of $\N$ and are neither inputs nor outputs of the XOR circuit.
The I/O CRN $\N$ is \emph{modular} since $\N_1$, $\N_2$, and $\N_3$ do not share any state species.
In fact, every I/O CRN produced by Construction~\ref{const:robust_circuit} is a modular join of NAND gates since combinatorial circuits are acyclic.

We now state the main theorem of the paper.
\begin{theorem}
    If $C_{n,m}$ is a combinatorial circuit, the constants $\bfm{\delta}=(\delta_1,\delta_2,\delta_3,\delta_4)\in(0,\infty)^4$ and $\tau > 0$ satisfy
    $\delta_2+\delta_3 < \delta_1 < \frac{1}{25}$,
    $\delta_2+\delta_3 < \frac{1}{100}$,
    and $\N=\N(C_{n,m},\bfm{\delta},\tau)$ is constructed according to Construction~\ref{const:robust_circuit}, then
    $\N\models_{\delta_1}^{\bfm{\delta}} \Phi(C_{n,m},\tau)$.
\end{theorem}
\begin{proof}
    This theorem immediately follows from the fact that $\N$ consists of a modular family of NAND gates and by Theorem~\ref{theorem:robust_nand} each individual NAND gate is robust.
    Thus, each NAND gate produces an output signal that is $\delta_1$-close to its appropriate binary value within $\frac{\tau}{d}$ time.
    Since $d$ is the depth of the circuit, the total propagation delay for the circuit is at most $\tau$.
\end{proof}

To demonstrate the robustness of these circuits, Figure~\ref{fig:54} also visualizes the output of the XOR circuit on a noisy input signal.
The simulation shows inputs that transition from low to high at different times, different levels, and different noise amplitudes.

\section{Robust Memory Components}
\label{sec:robust_memory}
Memory is essential to compute most algorithms, so limiting ourselves only to combinatorial circuits is too restricting.
The basic memory components of modern circuits are latches and flip flops, but these circuits are \emph{sequential} and depend on cyclic feedback to store data.
As a result, the techniques from the previous section do not apply, since joining our NAND gates together in a cyclic environment may cause them to send and receive signals that are not binary.
This can cause failure since the behavior of our NAND gate is undefined on non-binary inputs.

In this section, we show that I/O CRNs are capable of robustly simulating two common memory components.
In Section~\ref{sub:sr_latch}, we show that an SR latch can be robustly simulated by two NAND gates, and in Section~\ref{sub:d_latch}, we introduce a new I/O CRN design that robustly simulates a D latch.
A D latch is traditionally implemented using two SR latches; however, our I/O CRN construction uses fewer reactions than a single NAND gate.

\subsection{SR Latch}\label{sub:sr_latch}
The set-reset latch (SR latch) is a simple and commonly used memory element in digital circuits.
Composed of two NAND gates, the latch operates with two inputs, usually named $\Sbar$ and $\Rbar$, and has three stable states.
First, if $\Sbar$ is 0 and $\Rbar$ is 1, then the output $Q$ will be 1, \ie, $Q$ is \emph{set}.
Similarly, if $\Rbar$ is 0 and $\Sbar$ is 1, then the output $Q$ is 0, \ie, $Q$ is \emph{reset}.
If both $\Sbar$ and $\Rbar$ are 1, the output $Q$ maintains its previous value, \ie, $Q$ is \emph{held}.
A schematic diagram of the SR latch is shown in Figure~\ref{fig:srlatch}.
\begin{figure}
    \centering
    \hspace*{4em}
	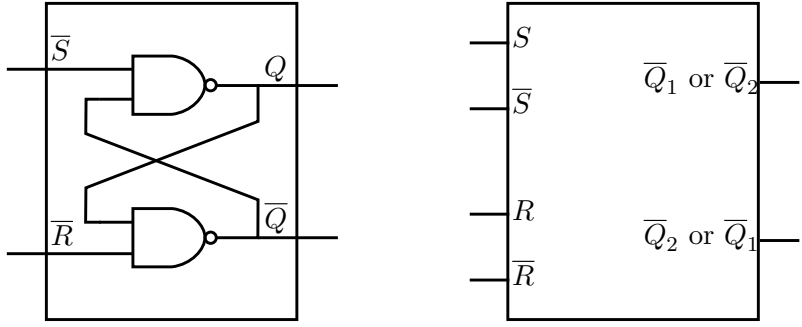
    \caption{\label{fig:srlatch}
        SR latch implemented with two NAND gates, and SR block diagram with labeled species inputs and outputs.
    }
\end{figure}

To show that this SR latch is robust, we begin by specifying its requirement.
We first define the set of input species, set of output species, and some useful predicates.
Given $\Sbar,\Rbar\in\bfm{S}$, we define the set of input species to be $\U = \{S,\Sbar,R,\Rbar\}\subseteq\bfm{S}$, and given $Q_1,\Qbar_2\in\bfm{S}$, we let the set of output species be $V=\{Q_1,\Qbar_1,Q_2,\Qbar_2\}\subseteq\bfm{S}$.
Given $\tau>0$, we also define the predicates
\begin{align}
    \phi_\text{set}(I)
        &\equiv
        \dbracket{\Sbar = 0}_{[t_1,t_1+\tau]}
        \land
        \dbracket{\Rbar = 1}_I\\
    \phi_\text{reset}(I)
        &\equiv
        \dbracket{\Rbar = 0}_{[t_1,t_1+\tau]}
        \land
        \dbracket{\Sbar = 1}_I\\
    \psi_a(I)
        &\equiv
        \dbracket{Q_1 = a}_{[t_1+\tau,t_2]}
        \land
        \dbracket{Q_2 = a}_{[t_1+\tau,t_2]},
\end{align}
for all intervals $I=[t_1,t_2]\in\bfm{I}(\tau)$.
Note that $\phi_\text{set}$ and $\phi_\text{reset}$ only require that $\Sbar=0$ and $\Rbar=0$ for the first $\tau$ time of $I$, but they require $\Rbar=1$ and $\Sbar=1$ for the entire interval $I$, respectively.
This allows inputs to transition between the set/reset state to the hold state while satisfying $\phi_\text{set}$/$\phi_\text{reset}$.

Given a $\tau>0$, we then define the SR latch requirement to be $\Phi_\text{SR}(\tau)=(\alpha,\phi)$ where the context assumption $\alpha$ is defined by
\begin{equation}
    \alpha(\bu,\V,h) \equiv \left[\V=\{Q_1,\Qbar_1,Q_2,\Qbar_2\}\text{ and }h=h_0\right],
\end{equation}
and the I/O requirement $\phi$ is defined by
\begin{equation}
    \phi(\bu,\bv)
        \equiv\big(\forall I\in\bfm{I}(\tau)\big)
            \left[
                \left(\phi_\text{set}(I)\rightarrow\psi_1(I)\right)
                \land
                \left(\phi_\text{reset}(I)\rightarrow\psi_0(I)\right)
            \right].
\end{equation}
Intuitively, the requirement $\Phi_\text{SR}$ requires that whenever $\Sbar=0$ and $\Rbar=1$ for at least $\tau$ time, then $Q=1$ within that time and remains there until $\Rbar\ne1$.
It also requires that if $\Sbar=1$ and $\Rbar=0$ for at least $\tau$ time, then $Q=0$ until $\Sbar$ is no longer 1.
A visualization of the input/output relationship is included in the timing diagram of Figure~\ref{fig:srtiming}.
\begin{figure}
    \centering
    \subfloat[]{
        \def\svgwidth{0.4\textwidth}
\begingroup%
  \makeatletter%
  \providecommand\color[2][]{%
    \errmessage{(Inkscape) Color is used for the text in Inkscape, but the package 'color.sty' is not loaded}%
    \renewcommand\color[2][]{}%
  }%
  \providecommand\transparent[1]{%
    \errmessage{(Inkscape) Transparency is used (non-zero) for the text in Inkscape, but the package 'transparent.sty' is not loaded}%
    \renewcommand\transparent[1]{}%
  }%
  \providecommand\rotatebox[2]{#2}%
  \newcommand*\fsize{\dimexpr\f@size pt\relax}%
  \newcommand*\lineheight[1]{\fontsize{\fsize}{#1\fsize}\selectfont}%
  \ifx\svgwidth\undefined%
    \setlength{\unitlength}{158.40000029bp}%
    \ifx\svgscale\undefined%
      \relax%
    \else%
      \setlength{\unitlength}{\unitlength * \real{\svgscale}}%
    \fi%
  \else%
    \setlength{\unitlength}{\svgwidth}%
  \fi%
  \global\let\svgwidth\undefined%
  \global\let\svgscale\undefined%
  \makeatother%
  \begin{picture}(1,0.75810946)%
    \lineheight{1}%
    \setlength\tabcolsep{0pt}%
    \put(0,0){\includegraphics[width=\unitlength,page=1]{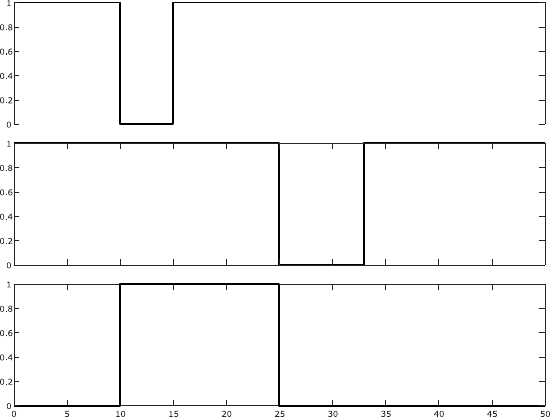}}%
    \put(0.09125695,0.68442721){\color[rgb]{0,0,0}\makebox(0,0)[lt]{\lineheight{0.08333334}\smash{\begin{tabular}[t]{l}$\Sbar$\end{tabular}}}}%
    \put(0.09337743,0.41560466){\color[rgb]{0,0,0}\makebox(0,0)[lt]{\lineheight{0.08333334}\smash{\begin{tabular}[t]{l}$\Rbar$\end{tabular}}}}%
    \put(0.09935234,0.16926704){\color[rgb]{0,0,0}\makebox(0,0)[lt]{\lineheight{0.08333334}\smash{\begin{tabular}[t]{l}$Q$\end{tabular}}}}%
    \put(0,0){\includegraphics[width=\unitlength,page=2]{SRFigureIdealNew.pdf}}%
  \end{picture}%
\endgroup%

     }
    \hspace*{2em}
    \subfloat[]{
        \def\svgwidth{0.4\textwidth}
\begingroup%
  \makeatletter%
  \providecommand\color[2][]{%
    \errmessage{(Inkscape) Color is used for the text in Inkscape, but the package 'color.sty' is not loaded}%
    \renewcommand\color[2][]{}%
  }%
  \providecommand\transparent[1]{%
    \errmessage{(Inkscape) Transparency is used (non-zero) for the text in Inkscape, but the package 'transparent.sty' is not loaded}%
    \renewcommand\transparent[1]{}%
  }%
  \providecommand\rotatebox[2]{#2}%
  \newcommand*\fsize{\dimexpr\f@size pt\relax}%
  \newcommand*\lineheight[1]{\fontsize{\fsize}{#1\fsize}\selectfont}%
  \ifx\svgwidth\undefined%
    \setlength{\unitlength}{158.40000029bp}%
    \ifx\svgscale\undefined%
      \relax%
    \else%
      \setlength{\unitlength}{\unitlength * \real{\svgscale}}%
    \fi%
  \else%
    \setlength{\unitlength}{\svgwidth}%
  \fi%
  \global\let\svgwidth\undefined%
  \global\let\svgscale\undefined%
  \makeatother%
  \begin{picture}(1,0.75810946)%
    \lineheight{1}%
    \setlength\tabcolsep{0pt}%
    \put(0,0){\includegraphics[width=\unitlength,page=1]{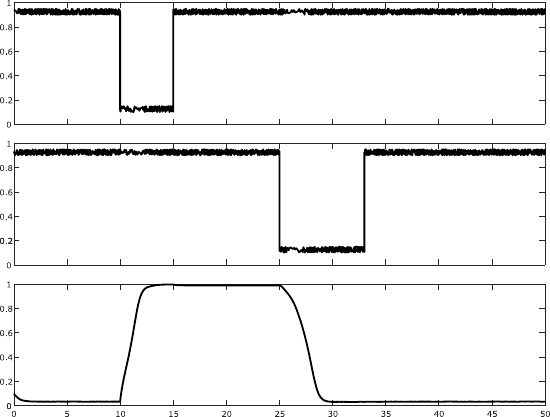}}%
    \put(0.09887419,0.66245706){\color[rgb]{0,0,0}\makebox(0,0)[lt]{\lineheight{0.08333334}\smash{\begin{tabular}[t]{l}$\Sbar$\end{tabular}}}}%
    \put(0.0982313,0.39832014){\color[rgb]{0,0,0}\makebox(0,0)[lt]{\lineheight{0.08333334}\smash{\begin{tabular}[t]{l}$\Rbar$\end{tabular}}}}%
    \put(0.09835379,0.15611707){\color[rgb]{0,0,0}\makebox(0,0)[lt]{\lineheight{0.08333334}\smash{\begin{tabular}[t]{l}$Q$\end{tabular}}}}%
  \end{picture}%
\endgroup%
  
    }
    \caption{\label{fig:srtiming}
        Ideal SR latch timing diagram along with an I/O CRN simulation of our SR latch with random noise}
\end{figure}

We now state the construction of the SR latch.
\begin{construction}\label{const:sr_latch}
    Given four species $\Sbar,\Rbar,Q_1,\Qbar_2$, a vector of strictly positive real numbers $\bfm{\delta}=(\delta_1,\delta_2,\delta_3,\delta_4)$, and $\tau>0$, define the CRN
    \[
        \text{SR}_{\bfm{\delta},\tau}(\Sbar,\Rbar,Q_1,\Qbar_2)
            = \N_1\sqcup\N_2,
    \]
    where $\N_1=\text{NAND}_{\bfm{\delta},\frac{\tau}{2}}(\Sbar,\Qbar_2,Q_1)$ and $\N_2=\text{NAND}_{\bfm{\delta},\frac{\tau}{2}}(\Rbar,Q_1,\Qbar_2)$.
\end{construction}

We now prove that our construction robustly satisfies $\Phi_\text{SR}$.
Our proof shows that the requirements of the two subcomponents suffice to prove the high-level requirement of the SR latch.

\begin{theorem}
    If $\bfm{\delta}=(\delta_1,\delta_2,\delta_3,\delta_4)\in(0,\infty)^4$ and $\tau > 0$ are constants satisfying $\delta_2+\delta_3 < \delta_1 < \frac{1}{25}$ and $\delta_2+\delta_3 < \frac{1}{100}$, then $\text{SR}_{\bfm{\delta},\tau}(\Sbar,\Rbar,Q_1,\Qbar_2)\models_{\delta_1}^{\bfm{\delta}} \Phi_\text{SR}(\tau)$.
\end{theorem}
\begin{proof}
    Assume the hypothesis, and let $\N=\text{SR}_{\bfm{\delta},\tau}(\Sbar,\Rbar,Q_1,\Qbar_2)$.
    We now let $\N_1=\text{NAND}_{\bfm{\delta},\frac{\tau}{2}}(\Sbar,\Qbar_2,Q_1)$ and $\N_2=\text{NAND}_{\bfm{\delta},\frac{\tau}{2}}(\Rbar,Q_1,\Qbar_2)$ be the I/O CRNs used to construct $\N$ from Construction~\ref{const:sr_latch}.
    By Theorem~\ref{theorem:robust_nand}, we know that
    \begin{align}
        \N_1&\models_{\delta_1}^{\bfm{\delta}} \Phi_\text{NAND}\left(\frac{\tau}{2}\right)\text{ and}\label{eq:n1_nand}\\
        \N_2&\models_{\delta_1}^{\bfm{\delta}} \Phi_\text{NAND}\left(\frac{\tau}{2}\right)\label{eq:n2_nand}
    \end{align}
    hold.
    We complete the proof by showing that these imply that $\N\models_{\delta_1}^{\bfm{\delta}} \Phi_\text{SR}(\tau)$.
    Note that $\Phi_\text{SR}$ can be easily split up into two parts.
    We first show that $\phi_\text{set}(I)\rightarrow\psi_1(I)$ holds, and then show that $\phi_\text{reset}(I)\rightarrow\psi_0(I)$ holds.

    Let $I=[t_1,t_2]\in\bfm{I}(\tau)$ be an interval such that $\phi_\text{set}(I)$ holds.
    Since $\dbracket{\Sbar = 0}$ holds for all $t\in[t_1,t_1+\tau]$,~\eqref{eq:n1_nand} tells us that $\dbracket{Q_1 = 1}$ for all $t\in[t_1+\frac{\tau}{2},t_1+\tau]$.
    Since $\dbracket{\Rbar = 1}$ and $\dbracket{Q_1 = 1}$ for all $t\in[t_1+\frac{\tau}{2},t_1+\tau]$,~\eqref{eq:n2_nand} tells us that $\dbracket{\Qbar_2 = 0}$ starting at time $t_1+\tau$.
    As a result, the output of $\dbracket{Q_1=1}$ and $\dbracket{\Qbar_2=0}$ is stable since the output of $\N_1$ will be held constant at 1 while one of its inputs is 0 and $\N_2$ will continue to output 0 while both its inputs are 1 which will be true until time $t_2$.
    Thus $\phi_1(I)$ holds for all $t\in[t_1+\tau,t_2]$.

    It remains to be shown that for all $I\in\bfm{I}(\tau)$, $\phi_\text{reset}(I)\rightarrow\psi_0(I)$ holds.
    Let $I=[t_1,t_2]\in\bfm{I}(\tau)$ be an interval such that $\phi_\text{reset}(I)$ holds.
    Since $\dbracket{\Rbar = 0}$ holds for all $t\in[t_1,t_1+\tau]$,~\eqref{eq:n2_nand} tells us that $\dbracket{\Qbar_2 = 1}$ for all $t\in[t_1+\frac{\tau}{2},t_1+\tau]$.
    Since $\dbracket{\Sbar = 1}$ and $\dbracket{\Qbar_2 = 1}$ for all $t\in[t_1+\frac{\tau}{2},t_1+\tau]$,~\eqref{eq:n1_nand} tells us that $\dbracket{Q_1 = 0}$ starting at time $t_1+\tau$.
    As a result, the output of $\dbracket{Q_1=0}$ and $\dbracket{\Qbar_2=1}$ is stable since the output of $\N_2$ will be held constant at 1 while one of its inputs is 0 and $\N_1$ will continue to output 0 while both its inputs are 1 which will be true until time $t_2$.
    Thus $\phi_0(I)$ holds for all $t\in[t_1+\tau,t_2]$
\end{proof}

Simulations show that the SR latch works even better than the theorem predicts.
Figure~\ref{fig:srlatch} shows its output with minor random noise and Figure~\ref{fig:sSRnoise} demonstrates how it handles significant random and sinusoidal noise.
\begin{figure}
    \centering
    \subfloat[]{
\begingroup%
  \makeatletter%
  \providecommand\color[2][]{%
    \errmessage{(Inkscape) Color is used for the text in Inkscape, but the package 'color.sty' is not loaded}%
    \renewcommand\color[2][]{}%
  }%
  \providecommand\transparent[1]{%
    \errmessage{(Inkscape) Transparency is used (non-zero) for the text in Inkscape, but the package 'transparent.sty' is not loaded}%
    \renewcommand\transparent[1]{}%
  }%
  \providecommand\rotatebox[2]{#2}%
  \newcommand*\fsize{\dimexpr\f@size pt\relax}%
  \newcommand*\lineheight[1]{\fontsize{\fsize}{#1\fsize}\selectfont}%
  \ifx\svgwidth\undefined%
    \setlength{\unitlength}{145.04944509bp}%
    \ifx\svgscale\undefined%
      \relax%
    \else%
      \setlength{\unitlength}{\unitlength * \real{\svgscale}}%
    \fi%
  \else%
    \setlength{\unitlength}{\svgwidth}%
  \fi%
  \global\let\svgwidth\undefined%
  \global\let\svgscale\undefined%
  \makeatother%
  \begin{picture}(1,0.8229087)%
    \lineheight{1}%
    \setlength\tabcolsep{0pt}%
    \put(0,0){\includegraphics[width=\unitlength,page=1]{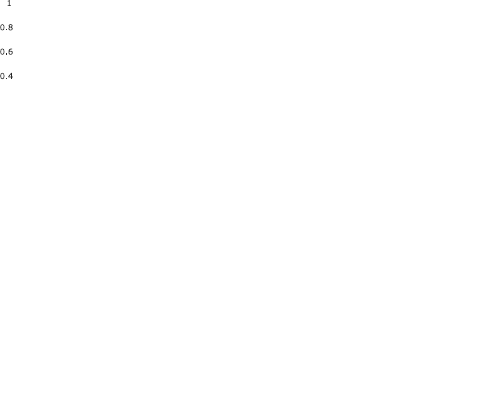}}%
    \put(0.05378004,0.64201265){\color[rgb]{0,0,0}\makebox(0,0)[lt]{\lineheight{0.08333334}\smash{\begin{tabular}[t]{l}$\Sbar$\end{tabular}}}}%
    \put(0,0){\includegraphics[width=\unitlength,page=2]{SRFigureNoise2New.pdf}}%
    \put(0.05378004,0.36485522){\color[rgb]{0,0,0}\makebox(0,0)[lt]{\lineheight{0.08333334}\smash{\begin{tabular}[t]{l}$\Rbar$\end{tabular}}}}%
    \put(0,0){\includegraphics[width=\unitlength,page=3]{SRFigureNoise2New.pdf}}%
    \put(0.05439999,0.1757308){\color[rgb]{0,0,0}\makebox(0,0)[lt]{\lineheight{0.08333334}\smash{\begin{tabular}[t]{l}$Q$\end{tabular}}}}%
    \put(0,0){\includegraphics[width=\unitlength,page=4]{SRFigureNoise2New.pdf}}%
  \end{picture}%
\endgroup%

    }
    \hspace*{2em    }
    \subfloat[]{
\begingroup%
  \makeatletter%
  \providecommand\color[2][]{%
    \errmessage{(Inkscape) Color is used for the text in Inkscape, but the package 'color.sty' is not loaded}%
    \renewcommand\color[2][]{}%
  }%
  \providecommand\transparent[1]{%
    \errmessage{(Inkscape) Transparency is used (non-zero) for the text in Inkscape, but the package 'transparent.sty' is not loaded}%
    \renewcommand\transparent[1]{}%
  }%
  \providecommand\rotatebox[2]{#2}%
  \newcommand*\fsize{\dimexpr\f@size pt\relax}%
  \newcommand*\lineheight[1]{\fontsize{\fsize}{#1\fsize}\selectfont}%
  \ifx\svgwidth\undefined%
    \setlength{\unitlength}{128.16190498bp}%
    \ifx\svgscale\undefined%
      \relax%
    \else%
      \setlength{\unitlength}{\unitlength * \real{\svgscale}}%
    \fi%
  \else%
    \setlength{\unitlength}{\svgwidth}%
  \fi%
  \global\let\svgwidth\undefined%
  \global\let\svgscale\undefined%
  \makeatother%
  \begin{picture}(1,0.90478694)%
    \lineheight{1}%
    \setlength\tabcolsep{0pt}%
    \put(0,0){\includegraphics[width=\unitlength,page=1]{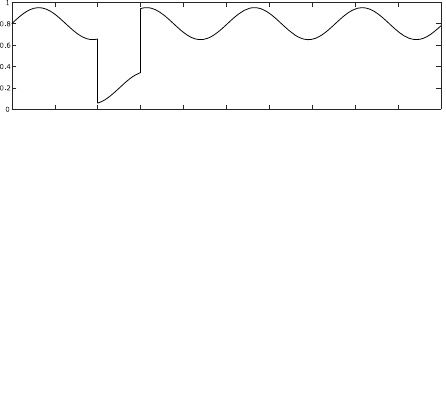}}%
    \put(0.06171055,0.4216465){\color[rgb]{0,0,0}\makebox(0,0)[lt]{\lineheight{0.08333334}\smash{\begin{tabular}[t]{l}$\Rbar$\end{tabular}}}}%
    \put(0,0){\includegraphics[width=\unitlength,page=2]{SRFigureNoise3New.pdf}}%
    \put(0.06057768,0.74407538){\color[rgb]{0,0,0}\makebox(0,0)[lt]{\lineheight{0.08333334}\smash{\begin{tabular}[t]{l}$\Sbar$\end{tabular}}}}%
    \put(0.05772268,0.15878482){\color[rgb]{0,0,0}\makebox(0,0)[lt]{\lineheight{0.08333334}\smash{\begin{tabular}[t]{l}$Q$\end{tabular}}}}%
    \put(0,0){\includegraphics[width=\unitlength,page=3]{SRFigureNoise3New.pdf}}%
  \end{picture}%
\endgroup%
  
    }
    \caption{\label{fig:sSRnoise}
        Simulations of the SR latch design with significant random and sinusoidal noise}
\end{figure}

\subsection{D Latch}\label{sub:d_latch}
Another commonly used memory element is the D latch.
Instead of using the traditional D latch design using four NAND gates, we provide a simpler construction using only four reactions.
The design is modeled closely after our NAND gate and uses the signal restoration algorithm of~\cite{cKlin16} to maintain the signals.
Before we give the construction, we first formally specify the requirement for a D latch.

Given species $D,E,Q\in\bfm{S}$ and $\tau>0$, define the set of input species be $\U=\{D,\Dbar,E,\Ebar\}\subseteq{\bfm{S}}$, let $V=\{Q,\Qbar\}\subseteq{\bfm{S}}$ be the set of output species, and for $a\in\{0,1\}$ let $\phi_a$ and $\psi_a$ be the predicates
\begin{align}
    \phi_a(I)
        &\equiv\dbracket{D = a\land E=1}_{[t_1,t_1+\tau]}\land
           \dbracket{D = a\lor  E=0}_{[t_1+\tau,t_2]}\\
    \psi_a(I)
        &\equiv\dbracket{Q = a}_{[t_1+\tau,t_2]}
\end{align}
for all $I=[t_1,t_2]\in\bfm{I}(\tau)$.
Then let $\Phi_\text{DL}(\tau)=(\alpha,\phi)$ be the requirement where the context assumption $\alpha$ is defined by
\begin{equation}
    \alpha(\bu,\V,h) \equiv \left[\V=\{Q,\Qbar\}\text{ and }h=h_0\right],
\end{equation}
and the I/O requirement $\phi$ is defined by
\begin{equation}
    \phi(\bu,\bv)
        \equiv\big(\forall I\in\bfm{I}(\tau)\big)
            \left[
                \left(\phi_0(I)\implies\psi_0(I)\right)
                \land
                \left(\phi_1(I)\implies\psi_1(I)\right)
            \right].
\end{equation}
Intuitively, the requirement $\Phi_\text{DL}$ says that whenever a set event occurs, \ie, when $D=a$ and $E=1$ for at least time $\tau$, then within $\tau$ time $Q$ converges to $a$ and remains there as long as either $D=a$ or $E=0$.
This is visualized in the timing diagram of Figure~\ref{fig:dlatchcrntiming}.
\begin{figure}[t]
    \begin{center}
    \subfloat[]{
\begingroup%
  \makeatletter%
  \providecommand\color[2][]{%
    \errmessage{(Inkscape) Color is used for the text in Inkscape, but the package 'color.sty' is not loaded}%
    \renewcommand\color[2][]{}%
  }%
  \providecommand\transparent[1]{%
    \errmessage{(Inkscape) Transparency is used (non-zero) for the text in Inkscape, but the package 'transparent.sty' is not loaded}%
    \renewcommand\transparent[1]{}%
  }%
  \providecommand\rotatebox[2]{#2}%
  \newcommand*\fsize{\dimexpr\f@size pt\relax}%
  \newcommand*\lineheight[1]{\fontsize{\fsize}{#1\fsize}\selectfont}%
  \ifx\svgwidth\undefined%
    \setlength{\unitlength}{133.55703014bp}%
    \ifx\svgscale\undefined%
      \relax%
    \else%
      \setlength{\unitlength}{\unitlength * \real{\svgscale}}%
    \fi%
  \else%
    \setlength{\unitlength}{\svgwidth}%
  \fi%
  \global\let\svgwidth\undefined%
  \global\let\svgscale\undefined%
  \makeatother%
  \begin{picture}(1,1.31225341)%
    \lineheight{1}%
    \setlength\tabcolsep{0pt}%
    \put(0,0){\includegraphics[width=\unitlength,page=1]{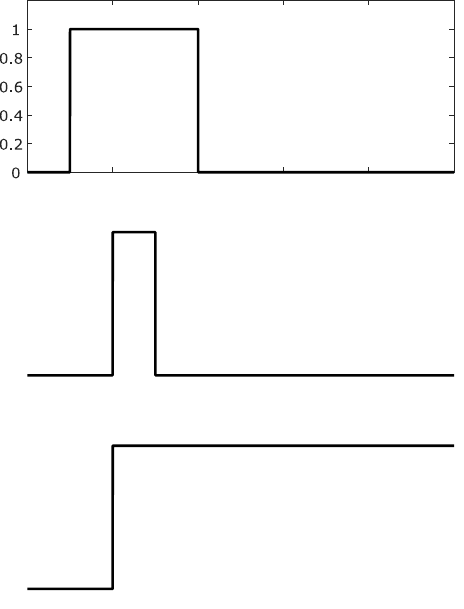}}%
    \put(0.07152345,1.19886873){\color[rgb]{0,0,0}\makebox(0,0)[lt]{\smash{\begin{tabular}[t]{l}$D$\end{tabular}}}}%
    \put(0.07152345,0.76085383){\color[rgb]{0,0,0}\makebox(0,0)[lt]{\smash{\begin{tabular}[t]{l}$E$\end{tabular}}}}%
    \put(0.07152345,0.30037664){\color[rgb]{0,0,0}\makebox(0,0)[lt]{\smash{\begin{tabular}[t]{l}$Q$\end{tabular}}}}%
    \put(0,0){\includegraphics[width=\unitlength,page=2]{DlatchIdealNew.pdf}}%
  \end{picture}%
\endgroup%

    }
    \hspace*{2em}
    \subfloat[]{
\begingroup%
  \makeatletter%
  \providecommand\color[2][]{%
    \errmessage{(Inkscape) Color is used for the text in Inkscape, but the package 'color.sty' is not loaded}%
    \renewcommand\color[2][]{}%
  }%
  \providecommand\transparent[1]{%
    \errmessage{(Inkscape) Transparency is used (non-zero) for the text in Inkscape, but the package 'transparent.sty' is not loaded}%
    \renewcommand\transparent[1]{}%
  }%
  \providecommand\rotatebox[2]{#2}%
  \newcommand*\fsize{\dimexpr\f@size pt\relax}%
  \newcommand*\lineheight[1]{\fontsize{\fsize}{#1\fsize}\selectfont}%
  \ifx\svgwidth\undefined%
    \setlength{\unitlength}{133.55702855bp}%
    \ifx\svgscale\undefined%
      \relax%
    \else%
      \setlength{\unitlength}{\unitlength * \real{\svgscale}}%
    \fi%
  \else%
    \setlength{\unitlength}{\svgwidth}%
  \fi%
  \global\let\svgwidth\undefined%
  \global\let\svgscale\undefined%
  \makeatother%
  \begin{picture}(1,1.30285652)%
    \lineheight{1}%
    \setlength\tabcolsep{0pt}%
    \put(0,0){\includegraphics[width=\unitlength,page=1]{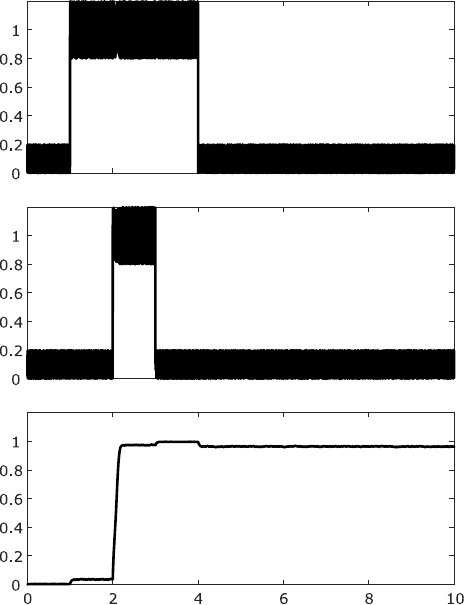}}%
    \put(0.06843668,1.17337546){\color[rgb]{0,0,0}\makebox(0,0)[lt]{\smash{\begin{tabular}[t]{l}$D$\end{tabular}}}}%
    \put(0.06843668,0.73536055){\color[rgb]{0,0,0}\makebox(0,0)[lt]{\smash{\begin{tabular}[t]{l}$E$\end{tabular}}}}%
    \put(0.06843668,0.27488233){\color[rgb]{0,0,0}\makebox(0,0)[lt]{\smash{\begin{tabular}[t]{l}$Q$\end{tabular}}}}%
  \end{picture}%
\endgroup%
  
    }
    \caption{\label{fig:dlatchcrntiming} CRN D latch timing diagram with random noise}
    \end{center}
\end{figure}

We now give the I/O CRN construction that satisfies the above requirement.

\begin{construction}\label{const:d_latch}
    Given three species $D,E,Q$, a vector of strictly positive real numbers $\bfm{\delta}=(\delta_1,\delta_2,\delta_3,\delta_4)$, and $\tau>0$, define the CRN
    \[
        \text{DL}_{\bfm{\delta},\tau}(D,E,Q)
            = (\U,\R,\S),
    \]
    where $\U=\{D,\Dbar,E,\Ebar\}$, $\S=\{Q,\Qbar\}$, and $\R$ consists of the four reactions
    \begin{align}
        D + E + \Qbar
            &\goesto{k}
            D + E + Q\\
        \Dbar + E + Q
            &\goesto{k}
            \Dbar + E + \Qbar\\
        2Q + \Qbar
            &\goesto{3k} 3Q\\
        2\Qbar + Q
            &\goesto{3k} 3\Qbar
    \end{align}
    where $k = 100\delta_4 + \frac{13}{\tau}$.
\end{construction}

Below is the final theorem of this paper showing that the above construction robustly satisfies its requirement.

\begin{theorem}\label{theorem:robust_d_latch}
    If $\bfm{\delta}=(\delta_1,\delta_2,\delta_3,\delta_4)\in(0,\infty)^4$ and $\tau > 0$ are constants satisfying $\delta_2+\delta_3 < \delta_1 < \frac{1}{25}$ and $\delta_2+\delta_3 < \frac{1}{100}$, then $\text{DL}_{\bfm{\delta},\tau}(\Sbar,\Rbar,Q_1,\Qbar_2)\models_{\delta_1}^{\bfm{\delta}} \Phi_\text{DL}(\tau)$.
\end{theorem}
\begin{proof}
    Assume the hypothesis and let $\N=(\U,\R,\S)=\text{DL}_{\bfm{\delta},\tau}(D,E,Q)$.
    We fix initial state $\xn\in\td^\S$ defined by $\xn(Q)=1$ and $\xn(\Qbar)=0$.
    Let $\bc=(\bu,\V,h)$ be a context that satisfies the context assumption $\alpha(\bc)$.
    Let $\bchat=(\buhat,\V,\hhat)$ be $(\delta_1,\delta_2)$-close to $\bc$, let $\xnhat$ be $\delta_3$-close to $\xn$, and let $\Nhat$ be $\delta_4$-close to $\N$.
    We fix $\bxhat\in C[\S]$ as the unique solution generated by $\Nhat$ in context $\bchat$ on the initial state $\xnhat$, and for convenience, we write $q(t)$ and $\qbar(t)$ to denote $\bxhat(t)(Q)$ and $\bxhat(t)(\Qbar)$, respectively.
    Now let $p=\xnhat(Q)+\xnhat(\Qbar)$.
    Since $\D{q}=-\D{\qbar}$, we know that $q(t)+\qbar(t) = p$ for all $t\in\td$.

    Let $I=[t_1,t_2]$ be an interval that satisfies $\phi_1(I)$.
    It is easy to show by bounding arguments similar to Theorem~\ref{theorem:robust_nand} that the inequality
    \[
        \D{q}\ge
            k\left[-\frac{p^3}{18} \left((3+d)^{3/2} + 9d\right) + (1-d)(1-\delta_1)^2(p-q) - 2(1+d)\delta_1\qbar\right]
    \]
    holds for all $t\in[t_1,t_1+\tau]$ where $d=\frac{\delta_4}{k}$.
    Similarly, we can easily show that
    \[
        \D{q}\ge
            k\left[3(1-d)q^2(p-q) - 3(1+d)q (p-q)^2 - 2\delta_1(1+d)q\right]
    \]
    holds for all $t\in I$.
    Thus by Lemmas~\ref{lemma:x_goesto_3/5} and~\ref{lemma:x_signal_restoration}, we see that $q(t)>p-\gamma$ and $\qbar(t)<\gamma$ for all $t\in[t_1+\tau,t_2]$ where $\gamma=\delta_1-\delta_2-\delta_3$.
    Thus $\bxhat$ is $\delta_1$-close to satisfying $\psi_1(I)$.

    By symmetry, if $I$ is an interval that satisfies $\phi_0(I)$, then $\bxhat$ is $\delta_1$-close to satisfying $\psi_1(I)$.
    Therefore $\N\models_{\delta_1}^{\bfm{\delta}}\Phi_\text{DL}(\tau)$.
\end{proof}

A simulation of the D latch operating on an input is visualized in Figure~\ref{fig:dlatchcrntiming}.
Again, random noise is added to demonstrate the robustness of the construction.  

\section{Discussion}\label{sec:discussion}
We have shown that any combinatorial circuit can be implemented by a robust input/output chemical reaction network.
By ``robust'' we mean that it tolerates bounded perturbations in the input signals, initial concentrations, reaction rate constants, and output measurements.
A key feature of our construction is that it preserves robustness under composition.
Thus, adding gates to a combinatorial circuit does not affect its robustness, however, it does increase the propagation delay if the new gates increase the depth of the circuit.
Preservation of robustness in this way allows designers to construct more complex circuits without needing to prove additional robustness theorems.

We have also shown that two sequential memory circuits can be implemented with robust I/O CRNs.
First, we showed that an SR latch can be constructed by composing two NAND gates together.
The proof of correctness relies solely on the proven requirements of the NAND gate subcomponents without any additional bounding arguments.
We also constructed a robust D latch which uses half the number of species and one-third the number of reactions of the SR latch construction.
This was a surprising reduction in complexity since traditional D latch designs use two SR latches (four NAND gates).

One drawback to our circuit design is that it does not inherently support hysteresis, and therefore circuits instantaneously react to changes in their input.
As a result, our construction fails on many common sequential circuits.
For example, a ring oscillator circuit constructed by connecting the output of a NAND gate to its own inputs ought to rapidly oscillate between 0 and 1.
However, it is easy to show that our implementation of such a circuit converges to an equilibrium state rather than rapidly oscillate.

Although some sequential circuits obviously fail, others can be constructed without issue.
For example, a negative edge-triggered D flip flop can be constructed using two D latches connected in a master-slave configuration.
In Figure~\ref{fig:dflipcrntiming}, we show a MATLAB Simbiology simulation of an I/O CRN design of this circuit composed of two D latches from Construction~\ref{const:d_latch}.
\begin{figure}
    \begin{center}
        \includegraphics[width=3.5in]{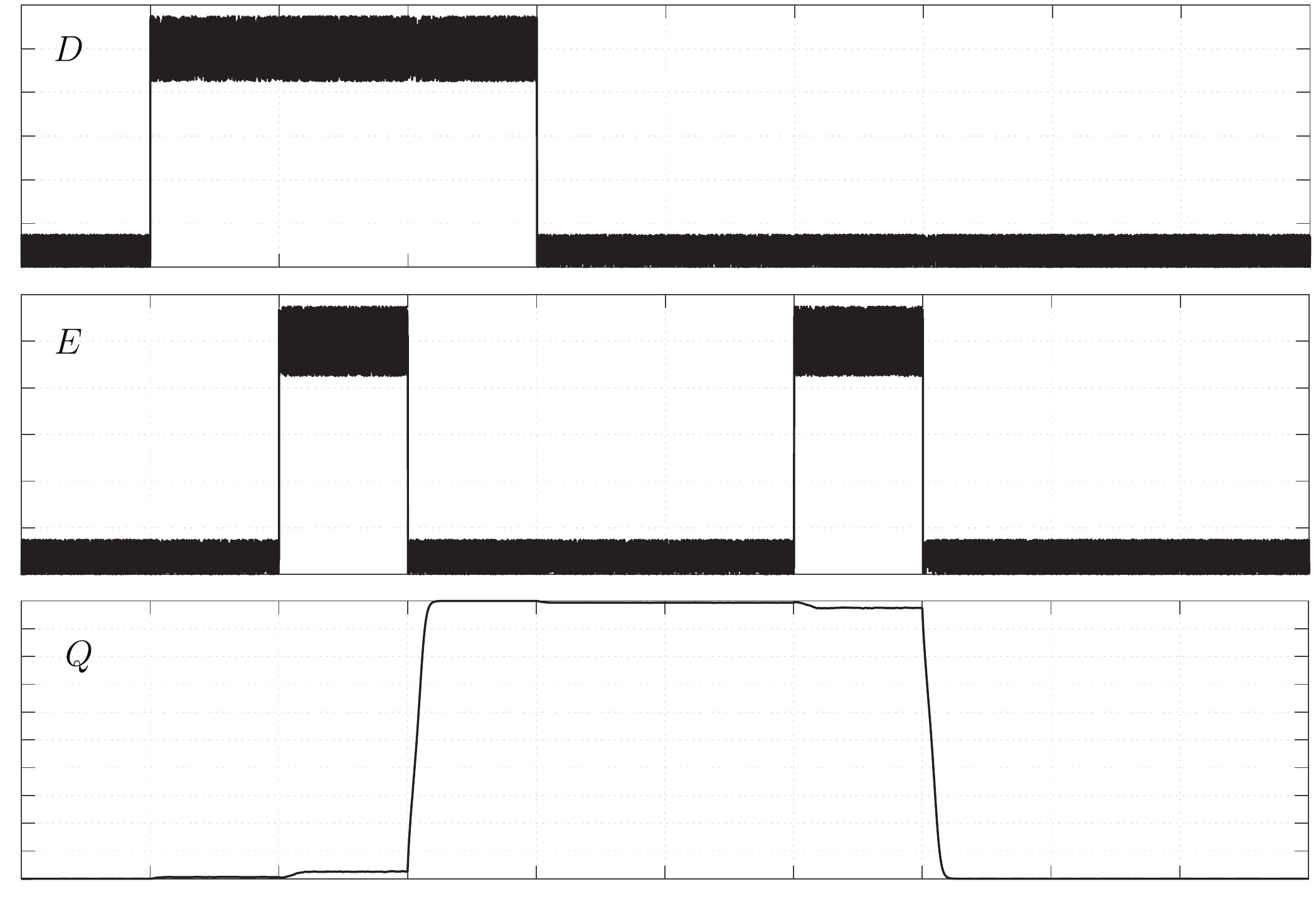}
        \caption{\label{fig:dflipcrntiming} CRN negative edge-triggered D flip-flop timing diagram with random noise}   
    \end{center}
\end{figure}
The simulations suggest that it works appropriately, and we suspect that techniques similar to those in Section~\ref{sec:robust_memory} can be used to show it is robust.
However, such proofs depend on properly stating the requirements of an edge-triggered flip flop, which is a natural next step to our research.

\section*{Acknowledgments}
We thank Jack Lutz and the Laboratory of Molecular Programming at Iowa State University for useful discussions.


\begin{thebibliography}{10}

\bibitem{jAris65}
Rutherford Aris.
\newblock Prolegomena to the rational analysis of systems of chemical
  reactions.
\newblock {\em Archive for Rational Mechanics and Analysis}, 19(2):81--99,
  1965.

\bibitem{arkin94}
A.~Arkin and J.~Ross.
\newblock Computational functions in biochemical reaction networks.
\newblock {\em Biophysical Journal}, 67(2):560 -- 578, 1994.

\bibitem{cBSJDTW17}
Stefan Badelt, Seung~Woo Shin, Robert~F. Johnson, Qing Dong, Chris Thachuk, and
  Erik Winfree.
\newblock A general-purpose {CRN}-to-{DSD} compiler with formal verification,
  optimization, and simulation capabilities.
\newblock In {\em Proceedings of the 23rd International Conference on {DNA}
  Computing and Molecular Programming}, Lecture Notes in Computer Science,
  pages 232--248, 2017.

\bibitem{beiki18}
Z.~Beiki, Z.~Zare Dorabi, and A.~Jahanian.
\newblock Real parallel and constant delay logic circuit design methodology
  based on the {DNA} model-of-computation.
\newblock {\em Microprocessors and Microsystems}, 61:217 -- 226, 2018.

\bibitem{jCard13}
Luca Cardelli.
\newblock Two-domain {DNA} strand displacement.
\newblock {\em Mathematical Structures in Computer Science}, 23(2):247--271,
  2013.

\bibitem{jCDSPCS13}
Yuan-Jyue Chen, Neil Dalchau, Niranjan Srinivas, Andrew Phillips, Luca
  Cardelli, David Soloveichik, and Georg Seelig.
\newblock Programmable chemical controllers made from {DNA}.
\newblock {\em Nature Nanotechnology}, 8(10):755--762, 2013.

\bibitem{oCSWB09}
Matthew Cook, David Soloveichik, Erik Winfree, and Jehoshua Bruck.
\newblock Programmability of chemical reaction networks.
\newblock In Anne Condon, David Harel, Joost~N. Kok, Arto Salomaa, and Erik
  Winfree, editors, {\em Algorithmic Bioprocesses}, Natural Computing Series,
  pages 543--584. Springer, 2009.

\bibitem{cDoty14}
David Doty.
\newblock Timing in chemical reaction networks.
\newblock In {\em Proceedings of the 25th Symposium on Discrete Algorithms},
  pages 772--784, 2014.

\bibitem{oElli17}
Samuel~J Ellis.
\newblock {\em Devices for safety-critical molecular programmed systems}.
\newblock PhD thesis, Iowa State University, 2017.

\bibitem{cEHKLLL14}
Samuel~J. Ellis, Eric~R. Henderson, Titus~H. Klinge, James~I. Lathrop, Jack~H.
  Lutz, Robyn~R. Lutz, Divita Mathur, and Andrew~S. Miner.
\newblock Automated requirements analysis for a molecular watchdog timer.
\newblock In {\em Proceedings of the 29th International Conference on Automated
  Software Engineering}, pages 767--778. ACM, 2014.

\bibitem{cFLBP17}
Fran{\c{c}}ois Fages, Guillaume Le~Guludec, Olivier Bournez, and Amaury Pouly.
\newblock Strong {T}uring completeness of continuous chemical reaction networks
  and compilation of mixed analog-digital programs", booktitle="proceedings of
  the 15th international conference on computational methods in systems
  biology.
\newblock pages 108--127. Springer International Publishing, 2017.

\bibitem{oFein79}
Martin Feinberg.
\newblock Lectures on chemical reaction networks, 1979.
\newblock \url{http://www.crnt.osu.edu/LecturesOnReactionNetworks}.

\bibitem{garg18}
Sudhanshu Garg, Shalin Shah, Hieu Bui, Tianqi Song, Reem Mokhtar, and John
  Reif.
\newblock Renewable time-responsive {DNA} circuits.
\newblock {\em Small}, page 1801470, 2018.

\bibitem{jGZWYZ17}
Lulu Ge, Zhiwei Zhong, Donglin Wen, Xiaohu You, and Chuan Zhang.
\newblock A formal combinational logic synthesis with chemical reaction
  networks.
\newblock {\em IEEE Transactions on Molecular, Biological and Multi-Scale
  Communications}, 3(1):33--47, March 2017.

\bibitem{oGuna03}
Jeremy Gunawardena.
\newblock Chemical reaction network theory for in-silico biologists, 2003.
\newblock \url{http://www.jeremy-gunawardena.com/papers/crnt.pdf}.

\bibitem{jHFLD09}
Thomas Hinze, Raffael Fassler, Thorsten Lenser, and Peter Dittrich.
\newblock Register machine computations on binary numbers by oscillating and
  catalytic chemical reactions modelled using mass-action kinetics.
\newblock {\em International Journal of Foundations of Computer Science},
  20(3):411--426, 2009.

\bibitem{jHjWeRo91}
Allen Hjelmfelt, Edward~D. Weinberger, and John Ross.
\newblock Chemical implementation of neural networks and {T}uring machines.
\newblock {\em Proceedings of the National Academy of Sciences},
  88(24):10983--10987, 1991.

\bibitem{hughes17}
Randall~A. Hughes and Andrew~D. Ellington.
\newblock Synthetic {DNA} synthesis and assembly: Putting the synthetic in
  synthetic biology.
\newblock {\em Cold Spring Harbor Perspectives in Biology}, 9(1), 2017.

\bibitem{jiang13}
Hua Jiang, Marc~D. Riedel, and Keshab~K. Parhi.
\newblock Digital logic with molecular reactions.
\newblock In {\em Proceedings of the 32nd International Conference on
  Computer-Aided Design}, pages 721--727. IEEE, 2013.

\bibitem{oKlin16}
Titus~H Klinge.
\newblock {\em Robust and Modular Computation with Chemical Reaction Networks}.
\newblock PhD thesis, Iowa State University, 2016.

\bibitem{cKlin16}
Titus~H. Klinge.
\newblock Robust signal restoration in chemical reaction networks.
\newblock In {\em Proceedings of the 3rd International Conference on Nanoscale
  Computing and Communication}, pages 6:1--6:6. ACM, 2016.

\bibitem{oKlLaLu15}
Titus~H. Klinge, James~I. Lathrop, and Jack~H. Lutz.
\newblock Robust biomolecular finite automata.
\newblock Technical Report 1505.03931, arXiv.org e-Print archive, 2015.

\bibitem{jack}
Titus~H. Klinge, James~I. Lathrop, and Jack~H. Lutz, 2016.
\newblock Work initially introduced in \cite{oKlin16} and will appear in a
  forthcoming extension of \cite{oKlLaLu15}.

\bibitem{oKraPar02}
Steven~G Krantz and Harold~R Parks.
\newblock {\em A primer of real analytic functions}.
\newblock Springer Science+Business Media, 2002.

\bibitem{jLYCP12}
Matthew~R. Lakin, Simon Youssef, Luca Cardelli, and Andrew Phillips.
\newblock Abstractions for {DNA} circuit design.
\newblock {\em Journal of The Royal Society Interface}, 9(68):470--486, 2012.

\bibitem{jMagn97}
Marcelo~O. Magnasco.
\newblock Chemical kinetics is {T}uring universal.
\newblock {\em Physical Review Letters}, 78(6):1190--1193, 1997.

\bibitem{jQiaWin11a}
Lulu Qian and Erik Winfree.
\newblock Scaling up digital circuit computation with {DNA} strand displacement
  cascades.
\newblock {\em Science}, 332(6034):1196--1201, 2011.

\bibitem{jSCWB08}
David Soloveichik, Matthew Cook, Erik Winfree, and Jehoshua Bruck.
\newblock Computation with finite stochastic chemical reaction networks.
\newblock {\em Natural Computing}, 7(4):615--633, 2008.

\bibitem{jSoSeWi10}
David Soloveichik, Georg Seelig, and Erik Winfree.
\newblock {DNA} as a universal substrate for chemical kinetics.
\newblock {\em Proceedings of the National Academy of Sciences},
  107(12):5393--5398, 2010.

\bibitem{srinivas17}
Niranjan Srinivas, James Parkin, Georg Seelig, Erik Winfree, and David
  Soloveichik.
\newblock Enzyme-free nucleic acid dynamical systems.
\newblock {\em Science}, 358(6369), 2017.

\bibitem{jYTMSN00}
Bernard Yurke, Andrew~J. Turberfield, Allen~P. Mills, Friedrich~C. Simmel, and
  Jennifer~L. Neumann.
\newblock A {DNA}-fuelled molecular machine made of {DNA}.
\newblock {\em Nature}, 406(6796):605--608, 2000.

\bibitem{jZhaSee11}
David~Yu Zhang and Georg Seelig.
\newblock Dynamic {DNA} nanotechnology using strand-displacement reactions.
\newblock {\em Nature Chemistry}, 3(2):103--113, 2011.

\bibitem{jZhaWin09}
David~Yu Zhang and Erik Winfree.
\newblock Control of {DNA} strand displacement kinetics using toehold exchange.
\newblock {\em Journal of the American Chemical Society}, 131(47):17303--17314,
  2009.

\end{thebibliography}
\end{document}